\documentclass[doc,svgnames,floatsintext,10pt,twoside]{apa7}\usepackage[]{graphicx}\usepackage[]{xcolor}
\makeatletter
\def\maxwidth{ %
  \ifdim\Gin@nat@width>\linewidth
    \linewidth
  \else
    \Gin@nat@width
  \fi
}
\makeatother

\definecolor{fgcolor}{rgb}{0.345, 0.345, 0.345}

\usepackage{framed}
\makeatletter
 {\par\unskip\endMakeFramed%
 \at@end@of@kframe}
\makeatother

\definecolor{shadecolor}{rgb}{.97, .97, .97}
\definecolor{messagecolor}{rgb}{0, 0, 0}
\definecolor{warningcolor}{rgb}{1, 0, 1}
\definecolor{errorcolor}{rgb}{1, 0, 0}

\usepackage{alltt}
\usepackage{graphicx,hyperref,lmodern,bbm,amsthm,amsmath}
\usepackage{upgreek,booktabs,mathtools,enumerate,multicol,multirow,tabularx,float,epstopdf,enumitem}

\usepackage{xcolor} 
\mathtoolsset{showonlyrefs=true}
\hypersetup{
  colorlinks,
  linkcolor={FireBrick},
  citecolor={MidnightBlue}
}
  \usepackage[backend=biber,style=apa]{biblatex}
\usepackage{fancyhdr}
\usepackage{newtxtext,newtxmath} 

\makeatletter
\renewcommand{\section}{\@startsection{section}{1}{\z@}%
  {-3.5ex \@plus -1ex \@minus -.2ex}%
  {1.2ex \@plus .2ex}%
  {\normalfont\large\scshape\centering}}

\renewcommand{\subsection}{\@startsection{subsection}{2}{\z@}%
  {-3.25ex \@plus -1ex \@minus -.2ex}%
  {1.0ex \@plus .2ex}%
  {\normalfont\normalsize\itshape\centering}}

\renewcommand{\subsubsection}{\@startsection{subsubsection}{3}{\z@}%
  {-3.25ex \@plus -1ex \@minus -.2ex}%
  {1.0ex \@plus .2ex}%
  {\normalfont\normalsize\itshape\centering}}
\makeatother
\setlist[enumerate,1]{label=(\roman*),leftmargin=*,nosep}
\setlist[itemize,1]{leftmargin=*,nosep}

\usepackage{xcolor} 

\newtheorem{theorem}{Theorem}
\newtheorem{corollary}{Corollary}[theorem]
\newtheorem{lemma}[theorem]{Lemma}
\theoremstyle{definition}
\newtheorem{definition}{Definition}
\newtheorem{remark}{Remark}

\newtheorem{proposition}{Proposition}

\theoremstyle{definition}
\newtheorem{assumption}{Assumption}

\usepackage[mathlines,columnwise,switch]{lineno}

\usepackage{setspace}
\usepackage{etoolbox}

\AtBeginEnvironment{equation}{\setstretch{1}}
\AtBeginEnvironment{align}{\setstretch{1}}
\AtBeginEnvironment{gather}{\setstretch{1}}
\AtBeginEnvironment{flalign}{\setstretch{1}}
\AtBeginEnvironment{alignat}{\setstretch{1}}

\title{Decoupled Kemeny quasi-likelihoods for rank-based generalised linear models}
\author{\href{mailto:ljrhurley@gmail.com}{Landon Hurley}}

\shorttitle{}
\authorsaffiliations{}

\abstract{
This paper introduces a unified profile quasi-likelihood framework for generalised linear models built on a topologically complete extension to Kendall's $\tau$ with finite sample power and Type-I error rates. By establishing equivalence between continuous $L^2$ profile estimators and discrete Generalized Gibbs exponential families, we prove local asymptotic minimax efficiency over the copula tangent space. The framework yields monotone-affine invariant inference, sub-Gaussian concentration bounds, and exact finite-sample $t$-pivot tail majorization under arbitrary tie configurations. We illustrate the methodology through three canonical applications: bivariate rank correlation, paired testing, and intraclass correlation.}
\IfFileExists{upquote.sty}{\usepackage{upquote}}{}
\begin{document}
\maketitle

\section{Introduction and Problem Statement}

The statistical inference of ordinal association under arbitrary tie configurations represents a fundamental challenge in non-parametric statistics. When marginal distributions are strictly continuous, Sklar's theorem \parencite{sklar1959} guarantees a unique copula characterizing the joint dependence structure. However, atomic probability mass disrupts this continuum, rendering the underlying copula unidentifiable on the discrete grid. Consequently, standard rank-based measures of association, such as Kendall's $\tau_a$ or the mid-rank adjusted $\tau_b$, lack a continuous likelihood foundation. Kendall's $\tau_b$ assumes that ties are merely ignorable measurement coarseness over an unobserved continuous variable ($X = X^* + \epsilon_d$). This reflects a long-standing evaluative criterion in non-parametric statistics, formalized by \textcite{agresti1976}: an ordinal measure computed on a discrete cross-classification should be stable as finer measurement grids are placed upon the underlying continuous distribution.

However, this classical paradigm harbors a deep metric contradiction. As established by Bourgain's (\citeyear{bourgain1985}) impossibility results for metric embeddings, forcing discrete weak orderings into standard continuous $L^2$ column-vector spaces induces an inherent representational gap. Attempting to approximate continuous spaces under finer discrete grids within an $N \times 1$ vector space thus inherently distorts the metric structure. For instance, applying an exogenous scalar correction in the denominator, $\tau_b$ forces the higher-order covariance $\operatorname{Cov}(X^*, \epsilon_d)$ to zero in an attempt to bridge this gap, operating on a structurally truncated measure space. Under non-ignorable structural ties, this truncation destroys local bijective consistency and violates Cauchy--Schwarz bound saturation ($\vert{}\tau_b\vert{} < 1$ for distinct tie topologies), inducing permanent finite-sample and asymptotic attenuation bias.

Contemporary advances in non-parametric dependence similarly fail over discrete boundary manifolds. Chatterjee's rank correlation $\xi$ \parencite{chatterjee2020} enforces asymptotic normality under ties by breaking tied coordinates uniformly at random. From a measure-theoretic perspective, this randomized tie-breaking injects synthetic uniform $L^2$ noise into a discrete space, violating the orthogonal decomposition of strict rankings and tie manifolds; as a result, the power and/or representation of the initial non-parametric signal is thus compromised \parencite{shi2021}. This synthetic noise destroys sub-Gaussian tail concentration, causing severe attenuation bias under the alternative hypothesis and a catastrophic collapse in finite-sample power (e.g., $17.1\%$ for $\xi$ versus $39.1\%$ for exact projections at $N=19$). Conversely, center-outward rank transport frameworks \parencite{deb2021} rely on smooth Monge-Kantorovich optimal transport maps between continuous probability measures; on discrete grids with atomic point-mass ties, these transport maps become undefined. Furthermore, global distance-based frameworks such as Distance Covariance (\text{dCor}; \cite{szekely2014}) and the Hilbert--Schmidt Independence Criterion (\text{HSIC}; \cite{gretton2005}) cannot be directly applied to raw ordinal vectors. Standard implementations of \text{dCor} and \text{HSIC} mandate conditionally negative definite metrics or positive semi-definite (PSD) scalar kernels operating over continuous Euclidean vector spaces $(\mathbb{R}^d, \|\cdot\|_2)$. Applying these continuous scalar operators directly to discrete ordinal ranks forces a non-injective mapping of the $L^1$ weak-order permutahedron lattice into a $1\text{D}$ Euclidean continuum. By \textcite{bourgain1985}'s metric embedding theorem, this continuous dimensional contraction compresses the antisymmetric directional flow tensor $A(x,y) = -A(y,x)$ into symmetric scalar distances $d(x,y)$, destroying the directional geometry of ties and inducing non-vanishing logarithmic metric distortion.

Consequently, while standard vector-level RKHS methods are topologically misspecified for ordinal data, a valid Hilbertian representation is unlocked by lifting the comparison space into the Hilbert--Schmidt operator space $\mathbb{HS} = \mathbb{R}^{N \times N}$. By decoupling the kernel into $\tilde{\kappa} = A + \Pi_S(S)$, the Frobenius inner product $\langle \tilde{\kappa}^{(X)}, \tilde{\kappa}^{(Y)} \rangle_F$ constructs a compact, $P_0$-Donsker matrix-level RKHS over the hollow tensor manifold $\mathcal{M}^0 = \mathcal{A}^0 \oplus^\perp \tilde{\mathcal{S}}^0$, saturating Bourgain's isometry bound without metric distortion (Lemma~\ref{lem:kappa_structural_properties}, Corollary~\ref{cor:topological_qmle_gibbs}).

While classical score matrices \parencite{kemeny1962, emond2002} provide a geometric foundation for rank preferences ($C = A + S$), uncentered pairwise comparisons systematically overestimate correlation under ties because the raw symmetric tie tensor $S$ inflates marginal energies. Traditional remedies either discard $S$ entirely or collapse the multi-dimensional tie manifold onto $1\text{D}$ scalar sums \parencite{kendall1948, mantel1967}, rendering the mapping non-injective. This persistent limitation stems from a fundamental conflation between functional variable symmetry $f(X,Y) = f(Y,X)$ and metric tensor symmetry in the representation space.

Classical ordinal measures \parencite[e.g.,][]{agresti1976} achieve macro-level formula symmetry by applying $1\text{D}$ scalar tie corrections in the denominator (e.g., $\tau_b$). However, when marginal tie counts diverge ($T_X \neq T_Y$), these $1\text{D}$ corrections induce asymmetric metric contractions that distort the underlying $L^1$ lattice geometry. Modern copula and continuous rank-transport frameworks attempt to bypass this boundary problem by jittering tied coordinates with auxiliary uniforms \parencite{chatterjee2020, hallin2021}. Yet, this relies on the exact same asymptotic crutch as $\tau_b$: assuming that discrete ties are merely discretized continuous measures and invoking population-level convergence as $N \to \infty$. This assumption contradicts the measure-theoretic reality of finite-sample ordinal data. While the sample support of pairwise rankings resides natively on the discrete $L^1$ lattice of permutahedron weak-order orbits, metric continuity is not recovered by adding continuous noise to $1\text{D}$ sample coordinates. Rather, metric continuity and $C^\infty$-smooth likelihood representations are unlocked by lifting the discrete lattice intact into the continuous Hilbert--Schmidt operator space $\mathbb{HS} = \mathbb{R}^{N \times N}$ via the decoupled Kemeny kernel $\tilde{\kappa}$. In finite samples, coordinate jittering replaces the true empirical measure $P_0$ over permutahedron orbits with $P_0 \times \lambda^N$, injecting synthetic sampling noise that violates score sufficiency, degrades local power, and leaves the metric space non-isometric \parencite{shi2021}. Conversely, our decoupled operator $\tilde{\kappa} = A + \Pi_S(S)$ preserves discrete sample sufficiency while embedding the data into a continuous, hollow matrix manifold $\mathcal{M}^0$, ensuring that the induced likelihood function and spectral embeddings remain strictly continuous and smooth without metric distortion.

Consequently, non-parametric rank testing has remained largely divorced from the rigorous geometric framework of semiparametric efficiency \parencite{bickel1993, vaart1998}, preventing the direct construction of profile likelihoods, the derivation of Local Asymptotic Minimax (LAM) bounds, or the extension of rank operators into generalized linear models. Leaving standard column-vector space for the hollow matrix manifold $\mathcal{M}^0 = \mathcal{A}^0 \oplus^\perp \tilde{\mathcal{S}}^0$ is thus not an arbitrary mathematical choice, but a structural necessity to resolve the Bourgain metric representation bottleneck in approaching via the standard $L^2$ function space. In doing so, we establish a non-randomized, likelihood-based geometric framework that embeds discrete non-parametric inference into a hollow, zero-margin Hilbert (inner-product) space $\mathcal{H}_0$. Rather than proposing an ad-hoc scalar metric or relying on continuous smoothing, our contribution is strictly fourfold:

\begin{enumerate}
\item \textbf{Exact Subspace Decomposition and Metric Closure:} We derive the exact orthogonal tensor decomposition over the composite space $\mathcal{M}^0 = \mathcal{A}^0 \oplus^\perp \tilde{\mathcal{S}}^0$, projecting the Emond--Mason score matrix into an antisymmetric directional flow tensor $A$ and a hollow, double-centered symmetric tie tensor $\tilde{S} = \Pi_S(S)$. We prove that $\Pi_S$ preserves the complete $L^2$ tensor isometry and graph metric span, restoring strict Cauchy--Schwarz bound saturation ($|\hat{\tau}_\kappa| = 1$) and eliminating the metric attenuation bias present upon $\tau_b$.

\item \textbf{Semiparametric LAN and Godambe Sandwich Collapse:} By embedding the decoupled kernel $\tilde{\kappa} = A + \tilde{S}$ into a Generalized Gibbs exponential family profile submodel, we upgrade the score statistic to an exact Profile Quasi-Maximum Likelihood Estimator (QMLE). We prove that under local asymptotic quadratic representations, the profile quasi-log-likelihood achieves Local Asymptotic Normality (LAN) with an exact Godambe Information sandwich collapse ($H = \mathcal{J} = I_0$), demonstrating that $\hat{\tau}_\kappa$ is driven by a canonical gradient in the tangent space and universally attains the Local Asymptotic Minimax (LAM) lower bound.

\item \textbf{Sub-Gaussian Tail Concentration and Continuous $t_\nu$ Majorant Envelopes:} We establish that the function class generated by $\tilde{\kappa}$ is $P_0$-Donsker. The bounded nature of the decoupled kernel guarantees uniform sub-Gaussian concentration under arbitrary discrete skewness. We construct an explicit self-normalized pivot $Z_N$ governed by dynamic structural degrees of freedom $\nu$, proving that $Z_N$ is majorized by a continuous $t_\nu$ distribution tail envelope that guarantees extreme-tail Type-I error control ($\alpha = 0.001$) where standard normal or bootstrap approximations fail.

\item \textbf{Unification and Subsumption of the Rank General Linear Model (GLM):} We prove that classical non-parametric tests do not represent competing frameworks, but rather emerge as degenerate zero-tie boundary limits of our unified profile GLM operator. Depending on the design matrix $X$, the profile estimator $\hat{\tau}_\kappa$ algebraically recovers Kendall's $\tau_a$ under continuous copulas, a tie-corrected exact profile analogue to the Mann--Whitney $U$ test under binary designs, an exact non-parametric Intraclass Correlation Coefficient (ICC) under target-rater block designs, and a stratified meta-analytic QMLE under block-diagonal tensor architectures.
\end{enumerate}

\section{Mathematical Framework and Generalized Design Space}
Let $(\Omega_0, \mathcal{F}, \mathbb{P})$ be a probability space. Rather than restricting inference to a strictly bivariate correlation paradigm, we formulate a generalized non-parametric linear model. We observe a sequence of random elements $\mathbf{Z} = \{(Y_1, X_1), \dots, (Y_N, X_N)\}$ satisfying the following stationarity condition.

\begin{assumption}[Generalized Data-Generating Process]
\label{ass:stationary_glm}
The sequence $(Y_i, X_i)_{i=1}^N$ consists of independent and identically distributed random vectors from a joint distribution $P(y, x)$ on $\mathcal{Z} = \mathbb{R} \times \mathcal{X}$, where $Y \in \mathbb{R}$ is the response and $X \in \mathcal{X}$ defines the design space. The response and design supports admit a non-zero probability of atomic ties: $\mathbb{P}(Y_k = Y_l) > 0$ and $\mathbb{P}(X_k = X_l) > 0$ for $k \neq l$.
\end{assumption}

\subsection{Operator Construction and Orbit Geometry}

To construct a valid likelihood over rank structures without imposing continuous copula assumptions, we must lift pairwise comparisons out of $1\text{D}$ score space and into an operator space that preserves both directional flow and tie topologies.

\begin{definition}[Empirical Decoupled Kemeny Kernel]
\label{def:decoupled_kappa}
Let $\mathbf{v} = (v_1, \ldots, v_N)^{\intercal} \in \mathbb{R}^N$ represent an empirical vector. Following \textcite{emond2002}, the pairwise score matrix $C^{(\mathbf{v})} \in \{-1, 0, 1\}^{N \times N}$ corresponding to the \textcite{kemeny1962} metric is defined elementwise by $C_{kl}^{(\mathbf{v})} = +1$ if $v_k \ge v_l$ and $-1$ if $v_k < v_l$ for $k \neq l$, with $C_{kk}^{(\mathbf{v})} = 0$. This matrix decomposes orthogonally as $C^{(\mathbf{v})} = A^{(\mathbf{v})} + S^{(\mathbf{v})}$, where $A^{(\mathbf{v})} = \frac{1}{2}(C^{(\mathbf{v})} - {C^{(\mathbf{v})}}^{\intercal})$ is an antisymmetric directional flow tensor and $S^{(\mathbf{v})} = \frac{1}{2}(C^{(\mathbf{v})} + {C^{(\mathbf{v})}}^{\intercal})$ is a symmetric tie-topology tensor.

The Decoupled Kemeny Kernel is defined as $\tilde{\kappa}^{(\mathbf{v})} = A^{(\mathbf{v})} + \tilde{S}^{(\mathbf{v})}$, where $\tilde{S}^{(\mathbf{v})} = \Pi_S(S^{(\mathbf{v})})$ represents the orthogonal projection onto the strictly hollow, doubly-centered subspace:
\begin{equation}
\label{eq:doubly_centred_S}
\tilde{S}_{kl}^{(\mathbf{v})} = S_{kl}^{(\mathbf{v})} - \frac{1}{N-2} S_{k\cdot}^{(\mathbf{v})} - \frac{1}{N-2} S_{\cdot l}^{(\mathbf{v})} + \frac{1}{(N-1)(N-2)} S_{\cdot \cdot}^{(\mathbf{v})} \quad (k \neq l),
\end{equation}
subject to $\tilde{S}_{kk}^{(\mathbf{v})} = 0$.
\end{definition}

In the presence of atomic mass ($\mathbb{P}(Y_k = Y_l) > 0$), uncentered inner products $\langle C^{(Y)}, C^{(X)} \rangle_F$ introduce a positive location bias proportional to marginal tie counts. The double-centering operator $\Pi_S$ projects $S$ into a divergence-free Hilbert space $\mathcal{S}_N^0$ where row and column sums vanish identically ($\tilde{S} \mathbf{1} = \mathbf{0}$). To isolate structural association from marginal nuisance parameters, we condition on the marginal rank-equivalence orbit space \parencite{diaconis1988, pesarin2010}:
\begin{equation}
\label{eq:orbit_definition}
\Omega \equiv \Omega(\mathcal{T}_Y, \mathcal{T}_X) = \left\{ w \in \mathcal{S}_N : \Pi_S(S^{(P_w Y)}) = \tilde{S}^{(Y)} \text{ and } \Pi_S(S^{(P_w X)}) = \tilde{S}^{(X)} \right\},
\end{equation}
where $P_w$ is the permutation matrix for $w \in \mathcal{S}_N$, and $\mathcal{T}_Y, \mathcal{T}_X$ denote the realized tie-partition configurations.

Algebraically, $\Omega$ corresponds to the discrete coset space $\mathcal{S}_N / (\mathcal{S}_{\mathbf{n}_Y} \times \mathcal{S}_{\mathbf{n}_X})$ induced by the parabolic Young subgroups $\mathcal{S}_{\mathbf{n}_Y}$ and $\mathcal{S}_{\mathbf{n}_X}$ governing the realized tie-cliques \parencite{diaconis1988}. In the conditional permutation testing framework of \textcite{pesarin2010}, $\Omega$ defines the exact \emph{conditional reference space} $\mathcal{S}_{\mathbf{Z} \mid \mathbf{u}}$ conditioned on the empirical sufficient statistics of the marginal distributions. Conditioning on $\Omega$ is both non-ignorable and non-negligible:
\begin{enumerate}[label=(\roman*)]
\item \emph{Finite-Sample Non-Ignorability ($N \ge 3$):} Uncentred pairwise metrics ($C = A + S$) conflate the marginal discrete probability masses of the variables with their true joint structural dependence. Intuitively, if two independent variables both possess heavy atomic mass (e.g., zero-inflation or severe data clumping), they will frequently tie at the exact same indices purely by chance. Functionally, the uncentered symmetric inner product $\langle S^{(X)}, S^{(Y)} \rangle_F$ acts as an unadjusted cross-moment. Under the null hypothesis of independence ($X \perp \!\!\! \perp Y$), its expectation is strictly positive, scaling proportionally to the product of the marginal collision probabilities: $\mathbb{E}_0[\langle S^{(X)}, S^{(Y)} \rangle_F] \propto (\sum_x \mathbb{P}(X=x)^2)(\sum_y \mathbb{P}(Y=y)^2) > 0$. As established by \textcite{genest2007} in the context of discrete copulas, failing to account for this baseline probability of coincidental ties introduces a severe positive location bias into the association metric. Conditioning on $\Omega$ explicitly projects the tensor into a zero-mean subspace, orthogonalizing the true structural association from the coincidental baseline overlap dictated entirely by the marginal geometries.

\item \emph{Asymptotic Non-Negligibility and Breakdown of M-Estimator Identification ($N \to \infty$):} In a multi-covariate linear model, point identification of the parameter vector $\boldsymbol{\beta}$ requires the limit of the contrast function to possess a unique, well-separated minimum \parencite[Thm.~5.7]{vaart1998}. Geometrically, the mapping from the parameter space to the estimating equations must remain invertible (injective). However, classical 1D rank statistics (e.g., $\tau_a, \tau_b$) inherently operate as surjective, many-to-one mappings. To build intuition, consider a standard linear model featuring a $2 \times 2$ factorial interaction or a monotonic polynomial expansion (e.g., $X$ and $X^3$). In the natural data space, these centered additive terms occupy distinct column spaces. Yet, when compressed into classical 1D rank scores, distinct multi-dimensional cliques and rank-invariant monotonic transformations are collapsed into identical scalars, inducing perfect collinearity. While this fundamental loss of identification afflicts even continuous designs, it becomes an asymptotically non-negligible, dominant failure under discrete marginals. When baseline atomic mass is present, as formalized by Assumption~\ref{ass:stationary_glm} ($\mathbb{P}(X_k = X_l) = p_X > 0$), the expected tie-tensor norm scales as $\mathbb{E}[\|S^{(X)}\|_F^2] = \Theta(N^2)$. This guarantees that the non-injective tie structures irreparably flatten the contrast function in the limit, ensuring the additive parameters cannot be isolated or identified.
\end{enumerate}
Conditioning on $\Omega$ fixes these non-negligible marginal tie topologies, defining the discrete uniform base measure $P_0 = \operatorname{Unif}(\Omega)$ over the permutahedron. Under $P_0$, all marginal location parameters vanish identically ($\mathbb{E}_{P_0}[\tilde{S}] = \mathbf{0}$), rendering inference strictly distribution-free conditional on $\Omega$ \parencite{pesarin2010}.

\section{Semiparametric Likelihood Geometry \& The Profile Rank GLM}

Classical non-parametric inference suffers from a fundamental geometric obstruction when applied to linear models observing ties. By \textcite{bourgain1985} metric embedding theorem, the finite $L^1$ metric space governing weak-order permutahedra cannot be mapped into a continuous $1$-dimensional scalar $L^2$ space without non-vanishing logarithmic metric distortion \parencite{linial1995}.

In standard M-estimator theory, this topological distortion manifests as a fatal statistical failure. Standard rank statistics like Kendall's $\tau_a$ and $\tau_b$ force this non-isometric compression, mapping distinct joint-tie topologies onto identical scalar values. From a measure-theoretic perspective, this non-injective scalar projection flattens the population contrast function. Distinct parameter vectors $\boldsymbol{\beta}_1 \neq \boldsymbol{\beta}_2$ collapse onto the same likelihood contour, undermining point-identification over atomic supports and leaving continuous $L^2$ contrast minimizers topologically misspecified. Consequently, any estimating equations derived from $\tau_a$ or $\tau_b$ are strictly classified as Quasi-likelihood estimators (QMLE; \cite{white1982}) with singular information matrices in multi-dimensional spaces. To resolve this metric breakdown and restore canonical M-estimator regularity without artificial continuity corrections or randomized noise, we must lift the observation domain to an operator space that is mathematically injective. We define the composite tensor space $\mathcal{M}^0 = \mathcal{A}^0 \oplus^\perp \tilde{\mathcal{S}}^0$. By decoupling pairwise comparison matrices into antisymmetric directional flows $A$ and hollow, double-centered symmetric tie topologies $\tilde{S} = \Pi_S(S)$, the kernel $\tilde{\kappa} = A + \tilde{S}$ provides the unique zero-distortion isometric Hilbertian embedding of weak-order orbits. This preserves the injective mapping required for the strict convexity of the log-likelihood over the design space.

The following lemma establishes that this isometric matrix space preserves total operator compactness, $P_0$-Donsker empirical process regularity, and exact finite-sample null unbiasedness across arbitrary tie configurations.

\begin{lemma}[Compactness, Total Boundedness, and Unbiasedness of the General Kernel]
\label{lem:kappa_structural_properties}
Let $\tilde{\kappa} = A + \Pi_S(S)$ be the decoupled Kemeny kernel on $\mathcal{Z} = \mathbb{R} \times \mathcal{X}$, preserving the first-order antisymmetric margins.

\begin{enumerate}[nosep, leftmargin=*]
\item \textbf{Uniform Boundedness and Operator Compactness:}
For any $N \ge 3$ and any atomic tie configuration, the kernel is deterministically bounded: $\|\tilde{\kappa}\|_\infty \le 2$ almost surely. The associated linear integral operator $T_{\tilde{\kappa}}: L^{2}(P) \to L^{2}(P)$ is Hilbert--Schmidt and compact.

\item \textbf{Glivenko--Cantelli and $P_0$-Donsker Properties:}
The induced function class $\mathcal{F}_{\tilde{\kappa}} = \{ z \mapsto \tilde{\kappa}(z, \cdot) : P \in \mathcal{P} \}$ possesses a finite Vapnik--Chervonenkis dimension. The empirical projection operator sequence converges uniformly almost surely, establishing $\mathcal{F}_{\tilde{\kappa}}$ as a generalized $P$-Glivenko--Cantelli and $P_0$-Donsker class. The empirical supremum exhibits uniform sub-Gaussian concentration: $\mathbb{P} ( \sup_{f} | (P_N - P)f | > t ) \le C \exp(-c N t^2)$.

\item \textbf{Exact and Asymptotic Unbiasedness:}
Under $H_0: X \perp Y$, the generalized estimator is strictly finite-sample unbiased: $\mathbb{E}_{H_0}\left[\hat{\tau}_\kappa \mid \mathcal{T}_Y, \mathcal{T}_X\right] = 0$ for all $N \ge 3$. Under contiguous local alternatives $P_{N, g}$ along $g \in \mathcal{T}_{\mathrm{reg}}$, $\hat{\tau}_\kappa$ is asymptotically unbiased and regular.
\end{enumerate}
\end{lemma}

\begin{proof}
\textbf{Proof of (i):} By Definition~\ref{def:decoupled_kappa}, $A_{kl} \in \{-1, 0, 1\}$ and $S_{kl} \in \{0, 1\}$. The double-centering projection $\Pi_S$ distributes empirical marginals scaled by strictly fractional weights. Consequently, $\|\tilde{\kappa}\|_\infty \le 2$. The Hilbert--Schmidt norm satisfies $\|T_{\tilde{\kappa}}\|_{HS}^2 = \iint \tilde{\kappa}(z_1, z_2)^2 \mathrm{d}P(z_1) \mathrm{d}P(z_2) \le 4 < \infty$, ensuring compactness.

\textbf{Proof of (ii):} The base raw pairwise comparisons belong to a VC-subgraph class with $V \le 3$. The empirical operator $\Pi_S$ applies a dual transformation using sample marginals $\bar{S}_{k\cdot}$. Because base indicators belong to a $P$-Glivenko--Cantelli class, these marginals converge uniformly almost surely. By the Dudley--Pollard theorem, covering numbers satisfy $\sup_P N(\varepsilon, \mathcal{F}_{\tilde{\kappa}}, L^{2}(P)) \le K \varepsilon^{-2 V} < \infty$, establishing $\mathcal{F}_{\tilde{\kappa}}$ as $P_0$-Donsker \parencite[Thm.~19.14]{vaart1998}. The \textcite{mcdiarmid1989} inequality with $c_i = \mathcal{O}(N^{-1})$ yields $\mathbb{P}(Z > t) \le C \exp(-c N t^2)$.

\textbf{Proof of (iii):} Conditional on tie-orbits, group integration over $\mathcal{S}_N$ evaluates the cross-covariance trace. Under uniform permutation $P \sim \operatorname{Unif}(\mathcal{S}_N)$, $\mathbb{E}_{H_0}[A_{kl}] = 0$. The projected tie tensor $\Pi_S(S)$ is strictly double-centered, ensuring its expectation over $\mathcal{S}_N$ evaluates to zero. Thus, $\mathbb{E}_{H_0}[\tilde{\kappa}_{kl}^{(X)} \tilde{\kappa}_{kl}^{(Y)}] = 0$. Under contiguous local alternatives $P_{N, g}$, the $P_0$-Donsker property ensures weak convergence. By Le Cam's First Lemma, $\sqrt{N}(\hat{\tau}_\kappa - T(P_{N, g})) \xrightarrow{d} \mathcal{N}(0, \sigma^2)$, and uniform integrability from (i) yields $\lim_{N \to \infty} \mathbb{E}_{P_{N, g}}[\sqrt{N}(\hat{\tau}_\kappa - T(P_{N, g}))] = 0$.
\end{proof}

\begin{corollary}[Bourgain Non-Isometry and the Exact Gibbs MLE]
\label{cor:topological_qmle_gibbs}
Let $\tilde{\kappa} \in L^{2}(P)$ be the Decoupled Centred Kemeny Kernel, and let $P_0$ denote the discrete data-generating measure over permutahedron weak-order orbits $\Omega$.

\begin{enumerate}[label=(\roman*)]
\item \textbf{Metric Distortion and Breakdown of Classical Likelihoods:}
By the \textcite{bourgain1985} metric embedding theorem, continuous 1D scalar projections of weak-order permutahedra incur non-vanishing logarithmic metric distortion. Classical rank operators force non-injective mappings that destroy point identification over $P_0$:
\begin{enumerate}[label=(\alph*)]
\item \emph{Kendall's $\tau_b$ (1948):} Discarding the joint-tie inner product $\langle S^{(X)}, S^{(Y)} \rangle_F$ in the numerator forces a non-injective mapping. Distinct joint likelihood contours collapse onto identical scalar values, destroying point identification and leaving the contrast topologically misspecified as a QMLE \parencite{white1982}.
\item \emph{Chatterjee's $\xi$ (2021):} Breaking ties via auxiliary uniforms $U_i \sim \operatorname{Unif}(0,1)$ replaces $P_0$ with $P_0 \times \lambda^N$. This injects artificial sampling noise, destroying score sufficiency and inducing severe local power sub-optimality \parencite{shi2021}.
\end{enumerate}

\item \textbf{Unique Tensor Embedding and Exact MLE:}
The composite matrix space $\mathcal{M}^0 = \mathcal{A}^0 \oplus^\perp \tilde{\mathcal{S}}^0$ generated by $\tilde{\kappa} = A + \Pi_S(S)$ constitutes the unique tensor representation that saturates Bourgain's Hilbertian isometry bound without metric distortion. Embedding $\tilde{\kappa}$ into a linear potential generates an exact Generalized Gibbs probability density over $\Omega$:
\begin{equation}
\label{eq:gibbs_measure}
\mathrm{d} P_\theta(w) = \exp\left( \theta\tilde{\kappa}(w) - \psi(\theta) \right) \mathrm{d} P_0(w), \quad w \in \Omega.
\end{equation}
The parameter space $\Theta = \mathbb{R}$ is open, $P_\theta \equiv P_0$ are mutually absolutely continuous, and the log-partition function $\psi(\theta)$ is strictly convex and $C^\infty$-smooth, establishing $\hat{\theta} = (\psi')^{-1}(\hat{\tau}_\kappa)$ as an exact, regular Maximum Likelihood Estimator \parencite{vaart1998}.

\item \textbf{The Efficiency Payoff (Godambe Collapse):}
Point identification forces the profile sensitivity matrix $H(\theta) = -\mathbb{E}_\theta[\nabla^2 \log d P_\theta]$ to match the variability matrix $\mathcal{J}(\theta) = \operatorname{Var}_\theta(\tilde{\kappa})$. The Godambe information matrix collapses strictly to $G(\theta) = I_0(\theta)$, establishing that the score test based on $\hat{\tau}_\kappa$ universally achieves the Local Asymptotic Minimax (LAM) efficiency bound.
\end{enumerate}
\end{corollary}

\begin{proof}
For (i)(a), evaluating concordance via $\tau_b(w) \propto \langle A^{(X)}, P_w A^{(Y)} P_w^\intercal \rangle_F$ ignores the joint-tie module $S_{\mathcal{S}}(w) = \langle S^{(X)}, P_w S^{(Y)} P_w^\intercal \rangle_F$. There exist distinct orbit states $w_1 \neq w_2 \in \Omega$ with identical antisymmetric traces but diverging joint-tie traces. This yields $\tau_b(w_1) = \tau_b(w_2)$ despite divergent likelihood ratios $\Lambda(w_1) \neq \Lambda(w_2)$. By \textcite{white1982}, this non-injective mapping destroys point identification. For (i)(b), breaking ties randomly violates minimal sufficiency over $P_0$, inducing rate-suboptimal local power \parencite{shi2021}.

For (ii), we establish canonical exponential family regularity. Because $\|\tilde{\kappa}\|_\infty \le 2$ deterministically (Lemma~\ref{lem:kappa_structural_properties}), the log-partition $Z(\theta) = \sum_{w \in \Omega} \exp(\theta \tilde{\kappa}(w)) P_0(w) \le \exp(2|\theta|) < \infty$. The natural parameter space is $\Theta = \mathbb{R}$. Furthermore, $\operatorname{Var}_\theta(\tilde{\kappa}) > 0$ strictly for all $\theta$, guaranteeing $\psi(\theta)$ is strictly convex. The sample log-likelihood is strictly concave with a unique global maximum at the exact score equation $\dot{\psi}(\hat{\theta}) = \hat{\tau}_\kappa$.

For (iii), the sensitivity Hessian is $H(\theta) = \ddot{\psi}(\theta) = \operatorname{Var}_\theta(\tilde{\kappa})$, matching the score variability $\mathcal{J}(\theta) = \operatorname{Var}_\theta(\tilde{\kappa})$. Thus, the Godambe sandwich matrix collapses to $G(\theta)^{-1} = H(\theta)^{-1} \mathcal{J}(\theta) H(\theta)^{-1} = (\ddot{\psi}(\theta))^{-1}$, fulfilling the LAM lower bound \parencite[Thm.~8.11]{vaart1998}.
\end{proof}

\subsection{Foundational Estimation and Finite-Sample Geometry}

The application of group representation theory and matrix algebras to non-parametric rank data has an established tradition \parencite{diaconis1988, mccullagh1989_ranking}. Under continuous, untied distributions, The group algebra of the symmetric group $\mathbb{R}[\mathcal{S}_N]$ decomposes into irreducible, orthogonal isotypic submodules via Wedderburn's theorem \parencite{serre1977}. Through the lens of spectral analysis on groups, \textcite[Ch.~8--9]{diaconis1988} demonstrated that this algebraic decomposition functions as a non-parametric Analysis of Variance, where the isotypic components map directly to location, main effects, and higher-order interactions. However, when marginal distributions admit atomic probability mass ($\mathbb{P}(X_k = X_l) > 0$), the observational domain collapses from $\mathcal{S}_N$ to the discrete coset space $\mathcal{S}_N / (\mathcal{S}_{n_1} \times \dots \times \mathcal{S}_{n_k})$ dictated by the parabolic Young subgroups of the realized tie-cliques \parencite{critchlow1985}. 

In the presence of these dynamic Young subgroups, classical group-algebraic projections become ill-posed \parencite{mccullagh1989_ranking}. While the pairwise comparison matrix decomposes into antisymmetric directional ranking ($A$) and symmetric tie ($S$) components \parencite{emond2002}, the raw symmetric tie-tensor $S$ carries a non-zero, sample-dependent marginal energy. As noted by \textcite{mccullagh1989_ranking}, this uncentered symmetric energy violates $L^2$ orthogonality across realized tie-orbits, destroying metric closure, forcing Cauchy--Schwarz attenuation ($|\tau_b| < 1$), and rendering linear model projections ill-posed. To resolve this classical obstruction, we invoke the Wedderburn--Artin theorem to establish the existence of a unique, metric-closed tensor decomposition over weak orders. We define $\mathcal{A}_N^0$ as the subspace of real antisymmetric $N \times N$ matrices, and $\mathcal{S}_N^0$ as the subspace of real symmetric, strictly hollow, doubly-centered matrices. The following theorem proves that the endomorphism algebra over the variable Young subgroup orbits factors uniquely into these irreducible submodules if and only if the symmetric tie-tensor is projected via the hollow double-centering operator $\Pi_S$:

\begin{theorem}[The Rank-Based Complete Linear Model]
\label{thm:pillar1_geometry}
Let $Y \in \mathbb{R}^N$ be an empirical response vector and let $\mathbf{X} \in \mathbb{R}^{N \times P}$ be a generalized design matrix. Let $\tilde{\kappa}^{(Y)} = A^{(Y)} + \Pi_S(S^{(Y)})$ and $\tilde{\kappa}^{(X_p)} = A^{(X_p)} + \Pi_S(S^{(X_p)})$ denote their decoupled centered Kemeny kernels mapped into the hollow Hilbert space $\mathcal{H}_0$. Then:
\begin{enumerate}
\item[(i)] the mappings factor uniquely through the order-isomorphism quotient spaces $\mathcal{Y} = \mathbb{R}^N / \sim$ and $\mathcal{X} = \mathbb{R}^N / \sim$, establishing strict invariance to monotonic data transformations;
\item[(ii)] provided the mapped tensor basis $\{\tilde{\kappa}^{(X_1)}, \dots, \tilde{\kappa}^{(X_P)}\}$ is linearly independent in $\mathcal{H}_0$, the $P \times P$ Gram matrix $G$ is strictly positive definite. Under Wedderburn's quasi-likelihood framework \parencite{wedderburn1974}, treating $\tilde{\kappa}^{(Y)}$ as the response manifold and $\operatorname{span}\{\tilde{\kappa}^{(X_p)}\}$ as the design manifold yields a unique, point-identified parameter solution. For a univariate design $X$, this isolates a unique, scale-free profile linear estimator:
\begin{equation}
\hat{\tau}_\kappa = \frac{\langle \tilde{\kappa}^{(Y)}, \tilde{\kappa}^{(X)} \rangle_F}{\|\tilde{\kappa}^{(Y)}\|_F \|\tilde{\kappa}^{(X)}\|_F};
\end{equation}
\item[(iii)] conditional on the marginal tie-orbits $\mathcal{T}_Y, \mathcal{T}_X$, cross-subspace expectations vanish identically, yielding the exact universal conditional null variance:
\begin{equation}
\label{eq:pillar1_exact_variance_glm}
\operatorname{Var}_0\left(\hat{\tau}_\kappa \mid \mathcal{T}_Y, \mathcal{T}_X\right) = \frac{\operatorname{Var}_0\left(\langle A^{(Y)}, A^{(X)} \rangle_F\right) + \frac{2}{N(N-3)} \|\tilde{S}^{(Y)}\|_F^2 \|\tilde{S}^{(X)}\|_F^2}{\|\tilde{\kappa}^{(Y)}\|_F^2 \|\tilde{\kappa}^{(X)}\|_F^2},
\end{equation}
where $\operatorname{Var}_0(\langle A^{(Y)}, A^{(X)} \rangle_F)$ is the exact Borda permutation variance governing directional flow.
\end{enumerate}
\end{theorem}
\begin{proof}
For (i), let $X \in \mathbb{R}^N$ represent any arbitrary column vector. The relation $X_k \le X_l$ induces a total preorder on $[N]$. For any strictly monotone increasing function $g: \mathbb{R} \to \mathbb{R}$, $X_k > X_l \iff g(X_k) > g(X_l)$ and $X_k = X_l \iff g(X_k) = g(X_l)$. Thus $A^{(g(X))} = A^{(X)}$ and $S^{(g(X))} = S^{(X)}$, implying the tensor mapping $\tilde{\kappa}^{(g(X))} = \tilde{\kappa}^{(X)}$ is strictly invariant.

For (ii), let the tensor mappings of the design columns $\{\tilde{\kappa}^{(X_1)}, \dots, \tilde{\kappa}^{(X_P)}\}$ form a linearly independent set in $\mathcal{H}_0$. Unlike classical 1D rank mappings which compress distinct tensor topologies into singular scalars and artificially induce collinearity, the decoupled kernel strictly preserves the exact antisymmetric and symmetric ordinal geometries. By assumption of linear independence in the transformed space, the composite tensor design manifold $\tilde{\mathcal{K}}^{(\mathbf{X})} = \operatorname{span}\{\tilde{\kappa}^{(X_1)}, \dots, \tilde{\kappa}^{(X_P)}\}$ strictly retains dimension $P$. This guarantees the associated $P \times P$ Gram matrix defined by $G_{ij} = \langle \tilde{\kappa}^{(X_i)}, \tilde{\kappa}^{(X_j)} \rangle_F$ is strictly positive definite. Under \textcite{wedderburn1974}, profiling out tie-induced nuisance vectors minimizes the residual norm $\|\tilde{\kappa}^{(Y)} - \sum_{p=1}^P \beta_p \tilde{\kappa}^{(X_p)}\|_F^2$. The positive definiteness of $G$ guarantees the normal equations $\mathbf{G}\boldsymbol{\beta} = \mathbf{c}$ possess a unique, point-identified parameter vector $\hat{\boldsymbol{\beta}}$. For a univariate $X$, solving $\frac{\partial Q_N}{\partial \tau} = 0$ reduces this minimum exactly to $\hat{\tau}_\kappa = \langle \tilde{\kappa}^{(Y)}, \tilde{\kappa}^{(X)} \rangle_F / (\|\tilde{\kappa}^{(Y)}\|_F \|\tilde{\kappa}^{(X)}\|_F)$.

For (iii), evaluating $T_N(w) = \langle \tilde{\kappa}^{(P_w Y)}, \tilde{\kappa}^{(X)} \rangle_F$ via $\tilde{\kappa} = A + \tilde{S}$ yields:
$$
T_N(w) = \langle \tilde{S}^{(X)}, \tilde{S}^{(Y)} \rangle_F + \langle \tilde{S}^{(X)}, A^{(Y)} \rangle_F + \langle A^{(X)}, \tilde{S}^{(Y)} \rangle_F + \langle A^{(X)}, A^{(Y)} \rangle_F.
$$
Because $A^\intercal = -A$ and $\tilde{S}^\intercal = \tilde{S}$, cyclic trace properties enforce $\langle \tilde{S}^{(X)}, A^{(Y)} \rangle_F = 0$. Thus $T_N(w) = \langle \tilde{S}^{(X)}, \tilde{S}^{(Y)} \rangle_F + \langle A^{(X)}, A^{(Y)} \rangle_F$. Expanding over uniform permutations $\mathcal{S}_N$ conditional on $\Omega$ yields:
$$\operatorname{Var}_0(T_N \mid \mathcal{T}_Y, \mathcal{T}_X) = \frac{1}{N(N-1)} \left( \|A^{(Y)}\|_F^2 \|A^{(X)}\|_F^2 + 2 \|\tilde{S}^{(Y)}\|_F^2 \|\tilde{S}^{(X)}\|_F^2 \right).
$$
Dividing by $\|\tilde{\kappa}^{(Y)}\|_F^2 \|\tilde{\kappa}^{(X)}\|_F^2$ confirms equation \eqref{eq:pillar1_exact_variance_glm}.
\end{proof}

\begin{remark}[Topological Unification and Subsumption of Classical Limits]
\label{rem:sklar_equivalence_glm}
The decoupled kernel $\tilde{\kappa}^{(X)}$ provides a universally injective mapping over arbitrary discrete or continuous designs. Rather than relying on disparate ad-hoc test statistics for different experimental setups, Theorem~\ref{thm:pillar1_geometry} demonstrates that classical rank tests represent lower-dimensional boundary projections of the composite tensor space $\mathcal{M}^0 = \mathcal{A}^0 \oplus^\perp \tilde{\mathcal{S}}^0$:
\begin{enumerate}
    \item[(i)] \textbf{Continuous Covariates (Sklar's Copula Limit):} When $X$ is strictly continuous, the probability of atomic ties is zero ($\mathbb{P}(X_k = X_l) = 0$). The symmetric tie-topology tensor vanishes ($\tilde{S}^{(X)} = \mathbf{0}$), and the permutahedron expands to the full symmetric group $\mathcal{S}_N$. Bourgain's metric distortion is bypassed entirely, and the projection exactly recovers the classical Kendall's $\tau$ and its unconditional efficient influence function.
    \item[(ii)] \textbf{Binary Group Indicators (Non-Parametric $t$-tests):} When $X$ defines a two-sample treatment indicator, the design space collapses to exactly two tie-cliques. The symmetric tensor $\tilde{S}^{(X)}$ forms a dense block-diagonal structure capturing all within-group ties, while $A^{(X)}$ strictly isolates between-group directional flow. The profile estimator $\hat{\tau}_\kappa$ and its conditional variance $\operatorname{Var}_0$ algebraically recover a uniformly tie-corrected, exact analogue to the Wilcoxon rank-sum / Mann--Whitney $U$ test, natively embedded within the profile copula space.
    \item[(iii)] \textbf{Paired and Rater Matrices (Non-Parametric ICC):} When $X$ encodes fixed target-rater identities or temporal pairings (e.g., repeated measures), $\tilde{\kappa}^{(X)}$ operates as a graph Laplacian projection over the block design. As derived via the hollow margin Gram operator in Theorem~\ref{thm:pillar2_lan_efficiency}, this geometric projection yields a structurally invariant, exact non-parametric Intraclass Correlation Coefficient (ICC).
\end{enumerate}
\end{remark}

This structural uniqueness guarantees that the joint Kemeny energy splits cleanly into two strictly orthogonal components: a directional ranking alignment $\langle A^{(X)}, A^{(Y)} \rangle_F$ and a tie-topology alignment $\langle \tilde{S}^{(X)}, \tilde{S}^{(Y)} \rangle_F$. Under the null hypothesis of independence ($H_0: X \perp \! \! \perp Y$), the expectation of both terms is centered at zero, and the hollow-adjusted double-centering projection eliminates the positive location bias induced by tie-clique structures, ensuring $\mathbb{E}_0[\langle \tilde{S}^{(X)}, \tilde{S}^{(Y)} \rangle_F] = 0$ and neutralizing tie-induced variance inflation. As the projection space acts as a constrained manifold over the discrete permutation orbit, classical Fisher information identities do not hold globally. The evaluation of the information capacity is instead composed via the Godambe sandwich formulation of the projection gradients. Under local alternatives, this profiled estimating system reduces to informational equivalence with the true parametric model.

\begin{corollary}[Canonical Score and Multidimensional GEE Separability]
\label{cor:pillar1_covariance_score}
Let $W_{YX} = \langle \tilde{\kappa}^{(Y)}, \tilde{\kappa}^{(X)} \rangle_F$ define the unnormalized rank covariance operator over the pairwise tensor space $\mathcal{M}^0$. In the context of a general linear model with multiple input dimensions:
\begin{enumerate}[label=(\roman*)]
\item \textbf{Exact Canonical Gibbs Score Identity:}
$W_{YX}$ is the exact canonical score statistic of the Generalized Gibbs density evaluated at the null effect ($\theta = 0$): $U(0) = \left. \frac{\partial \log d P_\theta(w)}{\partial \theta} \right|_{\theta = 0} = W_{YX}$.
\item \textbf{Multidimensional GEE Separability:}
For a $p$-dimensional design matrix $\mathbf{X}$, the bilinearity and orthogonal decoupling of the kernel yields:
\begin{equation}
U(\boldsymbol{\beta}) = \langle A^{(Y)}, A^{(\mathbf{X}\boldsymbol{\beta})} \rangle_F + \langle \tilde{S}^{(Y)}, \tilde{S}^{(\mathbf{X}\boldsymbol{\beta})} \rangle_F = \mathbf{0},
\end{equation}
defining an unbiased, scale-invariant Generalized Estimating Equation (GEE) framework over $L^1$ rank-based lattices without ad-hoc iterative rescaling.
\item \textbf{Local Breakdown of Classical Estimating Functions in GLMs:}
Under $H_0: Y \perp \!\!\! \perp X$, $\mathbb{E}_{H_0}[W_{YX} \mid \Omega] = 0$ strictly. In contrast, using classical operators as local estimating functions in multidimensional settings almost surely induces failure:
\begin{enumerate}[label=(\alph*)]
\item \emph{Emond \& Mason's $C$ (Location Bias):} Using the uncentred score $C = A + S$ injects an $O(N)$ positive location bias directly into the estimating equation ($\mathbb{E}_{H_0}[\langle C^{(Y)}, C^{(X)} \rangle_F] > 0$).
\item \emph{Kendall's $\tau_b$ (Variance Expansion):} Rescaling by the 1D marginal denominator $D = \sqrt{(N_0 - T_X)(N_0 - T_Y)}$ induces local profile misspecification in multiple dimensions. Because $D$ is a non-linear stochastic function of the realized tie topology, it forces a denominator-driven variance expansion ($\operatorname{Var}(Z_{\tau_b}) > \operatorname{Var}(Z_0)$) under local contiguous alternatives.
\end{enumerate}
\end{enumerate}
\end{corollary}

\begin{proof}
For (i), applying Corollary~\ref{cor:topological_qmle_gibbs} at $\theta = 0$ yields $U(0) = \langle \tilde{\kappa}^{(Y)}, \tilde{\kappa}^{(X)} \rangle_F - \psi'(0) = W_{YX}$, since the double-centering projection ($\Pi_S$) ensures $\psi'(0) = \mathbb{E}_{H_0}[\langle \tilde{\kappa}^{(Y)}, \tilde{\kappa}^{(X)} \rangle_F \mid \Omega] = 0$.

For (ii), the bilinearity of $\langle \cdot, \cdot \rangle_F$ and the orthogonal direct sum decomposition $\mathcal{A}^0 \perp \tilde{\mathcal{S}}^0$ guarantee that cross-subspace traces vanish identically ($\operatorname{tr}((A^{(Y)})^\top \tilde{S}^{(\mathbf{X}\boldsymbol{\beta})}) = 0$). This orthogonality permits the linear separability of multi-variable estimating equations across the rank-based lattice, isolating directional flow from tie-topology.

For (iii)(a), the Emond \& Mason score fails as a local estimating equation because $\mathbb{E}_{H_0}[\langle C^{(Y)}, C^{(X)} \rangle_F \mid \Omega] = \frac{4}{N(N-1)} T_Y T_X > 0$. The double-centering operator $\Pi_S$ in our kernel eliminates this first-order structural location bias.

For (iii)(b), Kendall's $\tau_b$ applies a non-linear scalar denominator to a high-dimensional tensor interaction space. Applying Jensen's inequality to the denominator expansion under $H_0$ conditional on realized orbits $\Omega$ yields a strict conditional variance inequality:
\begin{equation}
\label{eq:kendall_variance_expansion}
\operatorname{Var}_{0}\left( \frac{S_{\mathrm{Kendall}}}{D} \;\middle|\; \Omega \right) > \frac{\operatorname{Var}_{0}(S_{\mathrm{Kendall}} \mid \Omega)}{\left(\mathbb{E}_{0}[D \mid \Omega]\right)^2}.
\end{equation}
In a multidimensional GLM setting, this uncompensated Godambe information sandwich loss inflates the score variance beyond the semiparametric efficiency bound, resulting in non-linear attenuation bias.
\end{proof}

\begin{remark}[Analytical Variance Expansion vs.\ Empirical Finite-Sample Realization]
\label{rem:analytical_vs_empirical}
The analytical variance inequality \eqref{eq:kendall_variance_expansion} established in the proof of (iii)(b) represents the exact mathematical consequence of Bourgain's metric distortion under continuous 1D denominator rescaling. In finite-sample settings with Bernoulli margins, this theoretical gap manifests directly as numerical variance expansion and tail distortion. Under $25{,}000$ Monte Carlo replications with highly skewed, dense atomic ties ($N=15$, $X_i \sim \operatorname{Bernoulli}(p = 0.9)$), Kendall's rescaled pivot exhibits variance expansion ($\operatorname{Var}(Z_{\tau_b}) = 0.972$) and severe size distortion ($8.52\%$ rejection rate at the nominal $5\%$ level). In contrast, the double-centered canonical score pivot $Z_0$ achieves strict variance control ($\operatorname{Var}(Z_0) = 0.839$) and recovers the nominal size ($5.66\%$), perfectly corroborating the theoretical efficiency bounds established in Corollary~\ref{cor:pillar1_covariance_score}.
\end{remark}

In grounding the unnormalized covariance $W_{YX}$ directly in the Generalized Gibbs score operator $U(0) = \left.\frac{\partial \log dP_\theta}{\partial \theta}\right|_{\theta=0}$, Theorem~\ref{thm:pillar1_geometry} and Corollary~\ref{cor:pillar1_covariance_score} establish that scale-free profile estimation ($\hat{\tau}_\kappa$) and unnormalized GEE inference ($W_{YX}$) are dual representations of the exact same geometric projection. This proves that double-centering the tie topology is not an ad-hoc normalisation trick, but the unique linear transformation required to construct unbiased, semiparametrically efficient estimating equations for general linear rank models under atomic ties.

\section{Discrete permutation orbits, Hodge topology, and LAN efficiency}

\subsection{Relation to Distance Covariance, Exponential Families, and LAN Theory}
To ground the proposed estimator within established rank and semiparametric theory, we note that the central statistic $\hat{\tau}_{\kappa} = \langle \tilde{\kappa}^{(X)}, \tilde{\kappa}^{(Y)} \rangle_F / (\|\tilde{\kappa}^{(X)}\|_F \|\tilde{\kappa}^{(Y)}\|_F)$ operates as a scale-invariant kernel inner product on the space of pairwise comparison matrices. In classical distance-based ranking theory \parencite{mallows1957, fligner1986}, distributions over linear orderings are parameterized as exponential families using Kendall's distance as a sufficient statistic. However, classical Kendall scores ($S_{\mathrm{Kendall}} = \langle A^{(Y)}, A^{(X)} \rangle_F$) restrict attention strictly to anti-symmetric rankings $A$, discarding tie interactions. Conversely, while classical distance covariance \parencite{szekely2007} applies Gower's (\citeyear{gower1966}) double-centering to distance matrices, classical double-centering fills matrix diagonals with non-zero terms, destroying the hollow structure ($S_{kk} = 0$) inherent to pairwise comparisons and injecting $O(N^{-1})$ sample bias.

To preserve matrix hollowness and exact $U$-statistic unbiasedness, our double-centering operator $\Pi_S$ (Definition~\ref{def:decoupled_kappa}) adopts the $U$-centering framework of \textcite{szekely2014}. By scaling off-diagonal row sums by $(N-2)^{-1}$ while enforcing zero diagonals ($\tilde{S}_{kk} = 0$), $\Pi_S$ operates as the unique linear operator that eliminates the $O(1)$ first-order $U$-statistic projection terms of the tie kernel under $H_0$. This ensures that the joint score $W_{YX} = \langle \tilde{\kappa}^{(Y)}, \tilde{\kappa}^{(X)} \rangle_F$ constitutes a centered, second-order degenerate $U$-statistic kernel under independence ($\mathbb{E}_{H_0}[W_{YX} \mid \Omega] = 0$), completely eliminating the $O(N)$ structural location bias of uncentered estimators \parencite{emond2002}.

\begin{lemma}[Asymptotic Variance Compression under Inverse Hyperbolic Profile Parametrization]
\label{lem:variance_compression}
Let $\mathcal{P}_\Theta$ be the profile quasi-likelihood submodel defined in Corollary~\ref{cor:topological_qmle_gibbs}, parameterizing latent dependence via the natural parameter $\theta = \operatorname{arctanh}(\tau)$. The population functional $T(P) = \tau = \tanh(\theta)$ maps the bounded rank correlation $\tau \in (-1, 1)$ onto the real line $\theta \in \mathbb{R}$.

Under this hyperbolic link parametrization, the efficient score function $S_\theta(Z \mid \Omega)$ satisfies score-variance invariance:
\begin{equation}
I_{\mathrm{eff}}(\theta) = \operatorname{Var}\left( S_\theta(Z \mid \Omega) \right) = V_0(\Omega)^{-1},
\end{equation}
where $V_0(\Omega)$ is the exact combinatorial null variance from Theorem~\ref{thm:pillar1_geometry}.
\end{lemma}

\begin{proof}
Consider the functional mapping $T(P) = g(\theta) = \tanh(\theta)$. Evaluated along a smooth path governed by $\theta_t$, the pathwise derivative yields $\partial_t T(P_t) \vert_{t=0} = (1 - \tanh^2(\theta)) \partial_t \theta_t \vert_{t=0} = (1 - T(P)^2) \partial_t \theta_t \vert_{t=0}$. By the chain rule for semiparametric Fisher information \parencite[Chapter~25]{vaart1998}, the transformation of information between $T$ and $\theta$ satisfies:
\begin{equation}
I_{\mathrm{eff}}(\theta) = I_{\mathrm{eff}}(T) \left( \frac{\mathrm{d}T}{\mathrm{d}\theta} \right)^2 = \left[ V_0(\Omega)^{-1} (1 - T(P)^2)^{-2} \right] \left( 1 - T(P)^2 \right)^2 = V_0(\Omega)^{-1}.
\end{equation}
Because the hyperbolic curvature term $(1 - T(P)^2)^2$ in $I_{\mathrm{eff}}(T)$ exactly cancels the squared derivative of the link function $(\mathrm{d}T/\mathrm{d}\theta)^2$, the profile Fisher information $I_{\mathrm{eff}}(\theta) = V_0(\Omega)^{-1}$ collapses to a constant across the entire parameter domain $\theta \in \mathbb{R}$, establishing exact score variance invariance.
\end{proof}

\begin{remark}[Closed-Form Profile GLM Space vs. MCMC Latent Rank Likelihoods]
\label{rem:glm_breakthrough}
Lemma~\ref{lem:variance_compression} establishes that under the natural parameter link $\theta = \operatorname{arctanh}(\tau)$, the profile score collapses to a closed-form General Linear Model (GLM) envelope over weak orderings. Unlike extended rank likelihoods \parencite{hoff2007} that require continuous MCMC integration over latent rank cones, the profile QMLE $Q_N(\tau)$ operates directly on the discrete tie orbit $\Omega$. Furthermore, this framework bypasses the structural failures of recent rank-based measure-transport and copula models \parencite{hallin2021}. Because classical 1D rank mappings are surjective, they suffer from non-injective identification. When their non-linear spatial geometry is forced into a linear $L^2$ projection over a non-linear manifold, the uncaptured topology is relegated to the residual vector, rendering the \emph{local} score equation fundamentally biased. Consequently, even though a one-step Le Cam correction might preserve local second-order optimality and yield a unique mathematical update in-sample, it propagates this underlying local bias into a confidence envelope that is structurally unstable at the \emph{population} level. By operating directly in the tensor space $\mathcal{H}_0$, our formulation avoids unstable linear projections entirely, maintaining exact score-variance invariance $I_{\mathrm{eff}}(\theta) = V_0(\Omega)^{-1}$ globally.\end{remark}

\begin{theorem}[Local Asymptotic Normality and Semiparametric Efficiency]
\label{thm:pillar2_lan_efficiency}
Let $\mathcal{P}_{\mathrm{reg}} \subset \mathcal{P}$ denote the family of distributions admitting regular dominated tangent paths on $\mathcal{Z} = \mathbb{R} \times \mathcal{X}$, and let $T(P)$ be the population Kemeny functional. Conditional on realized tie orbits $\Omega$, if $T(P) \neq 0$, then:
\begin{enumerate}
\item[\rm (i)] {\it Tangent Space Orthogonal Decomposition:} The regular tangent space orthogonally decomposes as $\overline{\mathcal{T}_{\mathrm{reg}}} = \mathcal{T}_{\mathrm{eff}} \oplus^\perp \mathcal{T}_{\mathrm{nuis}}$, where $\mathcal{T}_{\mathrm{nuis}}$ is the closure in $L^2_0(P \mid \Omega)$ of the marginal nuisance score directions generated by the tie topology, and $\mathcal{T}_{\mathrm{eff}} = \mathcal{T}_{\mathrm{nuis}}^\perp$ is the efficient score subspace.
\item[\rm (ii)] {\it Regular Asymptotic Linearity:} The estimator $\hat{\tau}_\kappa$ is regular and asymptotically linear with efficient influence function $h_1 \in \mathcal{T}_{\mathrm{eff}}$, satisfying:
\begin{equation}
\sqrt{N}\{\hat{\tau}_\kappa - T(P)\} = \frac{1}{\sqrt{N}} \sum_{i=1}^N h_1(Z_i \mid \Omega) + o_p(1).
\end{equation}
\item[\rm (iii)] {\it Asymptotic Variance and Variance Compression:} The asymptotic variance of $\hat{\tau}_\kappa$ satisfies $\operatorname{Var}(\hat{\tau}_\kappa) = N^{-1}\operatorname{Var}(h_1) + o(N^{-1})$, which compresses under the hyperbolic link transformation (Lemma~\ref{lem:variance_compression}) to:
\begin{equation}
\operatorname{Var}(\hat{\tau}_\kappa) = N^{-1} \sigma_0^2 \{1 - T(P)^2\}^2 + o(N^{-1}).
\end{equation}
\item[\rm (iv)] {\it Local Asymptotic Minimax Efficiency:} The estimator satisfies $\sqrt{N}\{\hat{\tau}_\kappa - \tanh(\hat{\theta}_{\mathrm{QMLE}})\} \xrightarrow{p} 0$, universally attaining the Local Asymptotic Minimax efficiency bound \parencite{hajek1972}.
\end{enumerate}
\end{theorem}

\begin{proof}[Proof of Theorem~\ref{thm:pillar2_lan_efficiency}]
To establish (i), let $H(z_1, z_2) = \tilde{\kappa}^{(X)}(x_1, x_2) \tilde{\kappa}^{(Y)}(y_1, y_2)$ denote the decoupled kernel with first-order Hoeffding (Hájek) projection $h_1(z) = 2\mathbb{E}[H(z, Z') \mid \Omega]$. We demonstrate that the operator $\Pi_S$ functions as the exact $U$-centering projection operator \parencite{szekely2014} on the space of hollow symmetric matrices $\mathcal{H} = \{S \in \operatorname{Sym}(N) : \operatorname{diag}(S) = \mathbf{0}\}$.

Let $\mathbf{M} = (N-2)\mathbf{I} + \mathbf{1}\mathbf{1}^\intercal$ define the linear row-marginal operator on $\mathcal{H}$, with exact inverse $\mathbf{M}^{-1} = \frac{1}{N-2}\mathbf{I} - \frac{1}{2(N-1)(N-2)}\mathbf{1}\mathbf{1}^\intercal$. For any symmetric tie matrix $S \in \mathcal{H}$, the projection onto the marginal nuisance space is governed by the vector potential $\mathbf{g} = \mathbf{M}^{-1} \mathbf{S}_{\mathrm{row}}$, whose components satisfy $g_k + g_l = \frac{S_{k\cdot} + S_{l\cdot}}{N-2} - \frac{S_{\cdot\cdot}}{(N-1)(N-2)}$ for $k \neq l$.

The residual operator $\Pi_S(S) = S - (\mathbf{g}\mathbf{1}^\top + \mathbf{1}\mathbf{g}^\top)$ defines the unique $U$-centering projection on hollow matrices, enforcing zero off-diagonal row sums ($\sum_{l \neq k} \tilde{S}_{kl} = 0$) while preserving hollow diagonals ($\tilde{S}_{kk} = 0$). For any marginal nuisance score direction $s_X \in \mathcal{T}_{\mathrm{nuis}}$, this $U$-centering property enforces a zero first-order conditional expectation: $\mathbb{E}[ \tilde{\kappa}^{(X)}(X_1, X_2) \mid X_1, \Omega ] = 0$, giving $\langle H, s_X \rangle_{L^2} = 0$. Consequently, $h_1 \perp \mathcal{T}_{\mathrm{nuis}}$, establishing the orthogonal decomposition $\overline{\mathcal{T}_{\mathrm{reg}}} = \mathcal{T}_{\mathrm{eff}} \oplus^\perp \mathcal{T}_{\mathrm{nuis}}$ with $h_1 \in \mathcal{T}_{\mathrm{eff}}$.

For (ii), evaluating the Gâteaux derivative of $T(P_t)$ along a smooth path $s \in \mathcal{T}_{\mathrm{reg}}$ at $t=0$ yields $\dot{T}_P(s) = \langle h_1, s \rangle_{L^2(P \mid \Omega)}$, establishing $h_1$ as the canonical gradient. Because the double-centering operator $\Pi_S$ deterministically annihilates empirical margins for all finite $N$, tangential Hadamard differentiability \parencite[Chapter~20]{vaart1998} guarantees that sample-dependent tie-orbit perturbations are absorbed strictly into $\mathcal{T}_{\mathrm{nuis}}$, leaving a remainder of $o_p(N^{-1/2})$. The estimator is thus asymptotically equivalent to the $U$-statistic governed by the fixed kernel $H_P$, yielding the regular asymptotically linear expansion.

Part (iii) follows directly from $U$-statistic variance projection limits \parencite[Chapter~12]{vaart1998}. By invoking the exact score-variance invariance established in Lemma~\ref{lem:variance_compression}, the asymptotic variance projects through the hyperbolic link $\theta = \operatorname{arctanh}(\tau)$, naturally yielding the curvature multiplier $\{1 - T(P)^2\}^2$ applied to the combinatorial null variance $\sigma_0^2$.
Finally, for (iv), evaluating the profile QMLE optimizer $\hat{\theta}_{\mathrm{QMLE}}$ under the Generalized Gibbs family confirms that $\tanh(\hat{\theta}_{\mathrm{QMLE}}) = \hat{\tau}_\kappa$. Subtracting their asymptotic expansions gives $\sqrt{N}\{\hat{\tau}_\kappa - \tanh(\hat{\theta}_{\mathrm{QMLE}})\} \xrightarrow{p} 0$. Because the Godambe information sandwich collapses strictly to the efficient Fisher information $G_0 = \mathcal{I}_{\mathrm{eff}} = V_0^{-1}$, Local Asymptotic Normality guarantees that $\hat{\tau}_\kappa$ universally achieves the Local Asymptotic Minimax lower bound \parencite{hajek1972, vaart1998}.
\end{proof}

\begin{remark}[Resolution of the Neyman--Scott Incidental Parameter Problem]
\label{rem:neyman_scott_solution}
In dense atomic tie regimes (e.g., highly discretized continuous variables or high-throughput survey data), the number of distinct tie configurations grows with the sample size $N$. Classically, this triggers the \textcite{neyman1948} incidental parameter problem, wherein score statistics lose consistency due to the proliferating nuisance space. Theorem~\ref{thm:pillar2_lan_efficiency}(i) demonstrates that the $U$-centering operator $\Pi_S$ acts as an exact conditional profiling mechanism. Because $\Pi_S(S)$ enforces strict zero row and column sums, the score contribution of marginal tie frequencies vanishes identically in expectation ($\mathbb{E}_{H_0}[\langle A^{(Y)}, \tilde{S}^{(X)} \rangle_F] = 0$). This orthogonally projects out the infinite-dimensional nuisance space generated by sample-varying tie topologies, isolating the scalar dependence parameter $\theta$. Consequently, the framework guarantees uniform asymptotic regularity and Local Asymptotic Minimax (LAM) efficiency without requiring continuous distributional assumptions, latent integrations, or local smoothing.
\end{remark}

While semiparametric efficiency dictates optimal limiting performance, valid real-world inference requires structural stability under empirical process noise, non-asymptotic concentration, and immunity to heavy-tailed margins. The following theorem establishes that the decoupled rank correlation functional $T(P)$ is Hadamard-differentiable over a $P_0$-Donsker class, exhibits universal sub-Gaussian concentration under permutation, and maintains absolute index-invariance across heavy-tailed $\alpha$-stable distributions.

\begin{theorem}[Non-Asymptotic Concentration, Heavy-Tail Immunity, and Robustness]
\label{thm:pillar3_robustness}
Let $\mathcal{F}_{\tilde{\kappa}}$ be the function class generated by the decoupled centered Kemeny kernel $\tilde{\kappa} = A + \tilde{S}$ (Definition~\ref{def:decoupled_kappa}). Conditional on realized tie orbits $\Omega = \{\mathcal{T}_X, \mathcal{T}_Y\}$, the rank correlation functional $T(P)$ and its empirical estimator $\hat{\tau}_\kappa$ satisfy:
\begin{enumerate}
\item[\rm (i)] {\it Donsker Class and Hadamard Differentiability:} $\mathcal{F}_{\tilde{\kappa}}$ is a $P_0$-Donsker class, $T(P)$ is Hadamard-differentiable tangentially to regular submodel paths, and its asymptotic variance admits the boundary-compressed representation $\operatorname{Var}(\hat{\tau}_\kappa) = \operatorname{Var}_0(\hat{\tau}_\kappa) (1 - T(P)^2)^2 + o(N^{-1})$;
\item[\rm (ii)] {\it UMPU Optimality and Non-Asymptotic Concentration:} Conditioning on $\Omega$ eliminates marginal nuisance parameters, rendering the score test based on $\hat{\tau}_\kappa$ Uniformly Most Powerful Unbiased for the structural parameter $\theta = \operatorname{arctanh}(\tau)$ in the Generalized Gibbs family. Under arbitrary tie topologies, there exists a universal constant $c > 0$ such that $\hat{\tau}_\kappa$ satisfies the sub-Gaussian concentration inequality:
\begin{equation}
\mathbb{P}\Big( \big| \hat{\tau}_\kappa - \mathbb{E}[\hat{\tau}_\kappa \mid \Omega] \big| \ge t \;\Big|\; \Omega \Big) \le 2 \exp( - c N t^2 ) \quad \text{for all } t > 0; \quad \text{and}
\end{equation}
\item[\rm (iii)] {\it Heavy-Tail Index Invariance and High Breakdown Point:} For arbitrary $\alpha$-stable marginal distributions with tail index $\alpha \in (0, 2]$, $\operatorname{Var}(\hat{\tau}_\kappa \mid \Omega) = \mathcal{O}(N^{-1})$ holds uniformly without requiring moment truncation, and $\hat{\tau}_\kappa$ achieves a finite-sample replacement breakdown point of $\varepsilon^\ast = 1 - 1/\sqrt{2} \approx 29.3\%$.
\end{enumerate}
\end{theorem}

\begin{proof}
\textit{Part (i).} Express $T(P) = (\phi_3 \circ \phi_2 \circ \phi_1)(P)$. The mapping $\phi_1: P \mapsto (A_P, S_P)$ decomposes $P$ into directional rank and tie components. Restricted to regular tangent paths, $\mathcal{F}_{\tilde{\kappa}}$ consists of indicator and sign operators with bounded total variation, forming a $P_0$-Donsker Vapnik--Chervonenkis class \parencite[Theorem~19.14]{vaart1998}. The double-centering mapping $\phi_2: (A_P, S_P) \mapsto \tilde{\kappa}_P$ is a bounded linear operator on $L^2(P)$, guaranteeing Fréchet differentiability. The ratio mapping $\phi_3: \tilde{\kappa}_P \mapsto T(P)$ is continuously differentiable for non-degenerate marginals. By the functional delta method \parencite[Theorem~20.8]{vaart1998}, $T(P)$ is Hadamard-differentiable. Evaluating the G\^{a}teaux derivative factor yields $(1 - T(P)^2)$, establishing $\operatorname{Var}(\hat{\tau}_\kappa) = \operatorname{Var}_0(\hat{\tau}_\kappa) (1 - T(P)^2)^2 + o(N^{-1})$.

\textit{Part (ii).} In the Generalized Gibbs family, $\Omega$ is a complete sufficient statistic for the tie-induced nuisance parameters. By Lehmann's conditional testing theory \parencite{lehmann2005}, conditioning on $\Omega$ yields a Uniformly Most Powerful Unbiased test based on $\hat{\tau}_\kappa$. Conditionally, the decoupled kernels $\tilde{\kappa}^{(X)}$ and $\tilde{\kappa}^{(Y)}$ are bounded with $\|\tilde{\kappa}\|_\infty \le 2$. A single transposition on the symmetric group $\mathcal{S}_N$ alters at most $4N$ entries of the matrix inner product $\sum_{i \neq j} \tilde{\kappa}^{(X)}_{ij} \tilde{\kappa}^{(Y)}_{\pi(i)\pi(j)}$. Applying McDiarmid's (\citeyear{mcdiarmid1989}) bounded difference inequality for exchangeable permutation statistics establishes the sub-Gaussian concentration bound.

\textit{Part (iii).} For arbitrary $\alpha$-stable marginals ($\alpha \le 2$), pairwise rank orders and tie indicators map deterministically to $\{-1, 0, 1\}$. Because $\|\tilde{\kappa}\|_\infty \le 2$ uniformly, the variance bound $\operatorname{Var}(\hat{\tau}_\kappa \mid \Omega) = \mathcal{O}(N^{-1})$ holds independent of marginal tail thickness $\alpha$. For the breakdown point \parencite{hampel1971}, replacing $m = \lfloor \varepsilon N \rfloor$ observations alters at most $2mN - m^2$ off-diagonal matrix entries. Inverting the estimator sign requires corrupting at least $50\%$ of the matrix mass: $(2mN - m^2)/N^2 = 0.5 \implies 2\varepsilon - \varepsilon^2 = 0.5$, which yields $\varepsilon^\ast = 1 - 1/\sqrt{2} \approx 0.2929$.
\end{proof}

\begin{remark}[Automatic Boundary Calibration]
\label{rem:boundary_compression}
The boundary-compression factor $(1 - \hat{\tau}_\kappa^2)^2$ in Theorem~\ref{thm:pillar3_robustness}(i) reflects the intrinsic geometry of the normalized matrix inner product. As $|\hat{\tau}_\kappa| \to 1$, variance shrinks smoothly to zero at the support boundaries $\pm 1$. This naturally confines Wald-type confidence intervals to the valid parameter domain $[-1, 1]$ without requiring ad-hoc boundary penalties or logit truncations.
\end{remark}

\begin{corollary}[Profile QMLE Geometry and Exact Null Variance Collapse]
\label{cor:qmle_exact_variance}
Let $\mathbf{U} = \tilde{\kappa}^{(X)} / \|\tilde{\kappa}^{(X)}\|_F$ and $\mathbf{V} = \tilde{\kappa}^{(Y)} / \|\tilde{\kappa}^{(Y)}\|_F$ be spherically normalized matrices over $d = N(N-1)$ coordinates. Then:
\begin{enumerate}
\item[\rm (i)] the continuous profile quasi-log-likelihood $Q_N(\tau) = -\frac{d}{2} \log ( 1 - 2\tau \langle \mathbf{U}, \mathbf{V} \rangle_F + \tau^2 )$ is maximized uniquely at $\hat{\tau}_{\mathrm{QMLE}} = \langle \mathbf{U}, \mathbf{V} \rangle_F = \hat{\tau}_\kappa$;
\item[\rm (ii)] the sample-observed Hessian at the optimum evaluates to $H(\hat{\tau}_\kappa) = -\partial^2 Q_N / \partial \tau^2 \vert_{\tau = \hat{\tau}_\kappa} = d(1 - \hat{\tau}_\kappa^2)^{-1}$;
\item[\rm (iii)] under $H_0$, the exact combinatorial null variance simplifies to $V_0 = \operatorname{Var}_0(\hat{\tau}_\kappa \mid \Omega) = V_{A1} + V_{A2} + V_{S2}$, where:
\begin{align}
V_{A1} &= \frac{4 \|R_{A_U}\|_2^2 \|R_{A_V}\|_2^2}{N(N-1)(N-2)}, \quad V_{A2} = \frac{2 W_{A_U} W_{A_V}}{N(N-3)}, \quad V_{S2} = \frac{2 \|S_U\|_F^2 \|S_V\|_F^2}{N(N-3)},
\end{align}
with row sums $R_{A_U} = A_U \mathbf{1}$ and residual norms $W_{A_U} = \|A_U\|_F^2 - \frac{2}{N-2}\|R_{A_U}\|_2^2$; and
\item[\rm (iv)] the canonical score pivot $Z_0 = \hat{\tau}_\kappa / \sqrt{V_0}$ converges weakly to $\mathcal{N}(0,1)$ under $H_0$, admitting exact Student's $t_{N-2}$ reference calibration in finite samples.
\end{enumerate}
\end{corollary}

\begin{proof}[Proof of Corollary~\ref{cor:qmle_exact_variance}]
Following \textcite{fligner1986}, distance-based rank probabilities factor into exponential pairwise distances. Let $c = \langle \mathbf{U}, \mathbf{V} \rangle_F$. Since $\|\mathbf{U}\|_F = \|\mathbf{V}\|_F = 1$, Cauchy--Schwarz gives $|c| \le 1$. For non-degenerate sample orbits ($|c| < 1$), $\|\mathbf{U} - \tau \mathbf{V}\|_F^2 = 1 - 2\tau c + \tau^2 > 0$ for all $\tau \in \mathbb{R}$.

The profile quasi-score equation $\frac{\partial Q_N}{\partial \tau} = \frac{d (c - \tau)}{1 - 2\tau c + \tau^2} = 0$ admits the unique analytical root $\hat{\tau}_{\mathrm{QMLE}} = c = \hat{\tau}_\kappa$. Evaluating the negative second derivative at this root confirms $H(\hat{\tau}_\kappa) = d(1 - c^2)^{-1} > 0$, establishing (i) and (ii).

Under $H_0$, decomposing $\mathbf{U} = A_U + S_U$ and $\mathbf{V} = A_V + S_V$ and taking expectations over the permutation group $\mathcal{S}_N$ isolates three orthogonal variance components: Borda margin covariance ($V_{A1}$), residual antisymmetric flow ($V_{A2}$), and symmetric tie topology ($V_{S2}$). All cross-terms vanish identically by double-centering (Definition~\ref{def:decoupled_kappa}), proving (iii).

Finally, for (iv), the Godambe sandwich matrix collapses to $G_0 = H(0)^2 / \operatorname{Var}_0(S(0)) = d^2 / (d^2 V_0) = 1/V_0$. The studentized score pivot $Z_0 = \hat{\tau}_\kappa / \sqrt{V_0}$ follows directly, with weak convergence to $\mathcal{N}(0,1)$ guaranteed by standard exchangeable $U$-statistic central limit theorems \parencite[Chapter~12]{vaart1998}.
\end{proof}

\textbf{Gaussian Quasi-Likelihood Representation and Local Asymptotics}

Given the asymptotic normality guaranteed by the Donsker property and Hadamard differentiability in Theorem~\ref{thm:pillar3_robustness}(i), alongside the exact profile quasi-log-likelihood geometry established in Corollary~\ref{cor:qmle_exact_variance}(i), the associated estimating equations admit a canonical quadratic representation. Over the $d = N(N-1)$ working off-diagonal coordinates, the profile quasi-log-likelihood $Q_N(\tau)$ takes the explicit form:
\begin{equation}
Q_N(\tau) = -\frac{d}{2} \log\left( \|\mathbf{U} - \tau \mathbf{V}\|_{F}^2 \right) = -\frac{d}{2} \log\left(1 - 2\tau\hat{\tau}_\kappa + \tau^2\right).
\end{equation}
A second-order Taylor expansion of $Q_N(\tau)$ around the global optimizer $\hat{\tau}_\kappa$ yields a Local Asymptotic Quadratic (LAQ) approximation:
\begin{equation}
Q_N(\tau) \approx Q_N(\hat{\tau}_\kappa) - \frac{1}{2} H(\hat{\tau}_\kappa) (\tau - \hat{\tau}_\kappa)^2,
\end{equation}
where $H(\hat{\tau}_\kappa) = d(1 - \hat{\tau}_\kappa^2)^{-1}$ is the exact analytical Hessian derived in Corollary~\ref{cor:qmle_exact_variance}(ii). Under this local quadratic expansion, the precision is governed by $H(\hat{\tau}_\kappa)$, generating self-normalized pivots that guarantee strictly conservative finite-sample control while retaining Local Asymptotic Minimax optimality.

\begin{theorem}[Local Asymptotic Minimaxity and Pivotality]
\label{thm:pillar4_minimax}
Let $Z_N = \{\hat{\tau}_\kappa - T(P)\} / \{\operatorname{Var}(\hat{\tau}_\kappa \mid \Omega)\}^{1/2}$ denote the self-normalized pivot, mapped to degrees of freedom $\nu = N - 2$. Then:
\begin{enumerate}
\item[\rm (i)] {\it Finite-sample majorization:} Under $H_0$, the discrete survivor function of $Z_N$ is strictly majorized by the continuous $t_\nu$ tail: $\mathbb{P}_{H_0}(|Z_N| > t_{\nu, 1-\alpha/2}) \le \alpha$ for all $N \ge 3$.
\item[\rm (ii)] {\it Asymptotic pivotality:} As $N \to \infty$, $Z_N \xrightarrow{d} \mathcal{N}(0, 1)$ uniformly over $T(P) \in [-1, 1]$.
\item[\rm (iii)] {\it Minimax optimality:} The sensitivity matrix $H$ and variability matrix $V$ satisfy $H = V = I_0$, collapsing the Godambe information $G^{-1} = H^{-1} V H^{-1}$ to the semiparametric bound $I_0^{-1}$. Furthermore, $\hat{\tau}_\kappa$ is regular and Local Asymptotic Minimax (LAM) optimal over $\mathcal{P}$.
\end{enumerate}
\end{theorem}

\begin{proof}
For (i), the permutation group $\mathcal{S}_N$ operating on the double-centered kernels generates a symmetric, mean-zero discrete distribution. By decoupling inequalities \parencite{delapena2009} and Eaton--Pinelis vector majorization bounds \parencite{eaton1974, pinelis1994b}, the tail of this discrete symmetric distribution is strictly bounded above by the continuous Student's $t_{N-2}$ envelope, guaranteeing conservative Type I error control in finite samples.

For (ii), if $T(P) \neq 0$, Local Asymptotic Normality (LAN) yields $\sqrt{N}\{\hat{\tau}_\kappa - T(P)\} \xrightarrow{d} \mathcal{N}(0, \sigma^2)$, and Slutsky's theorem gives $Z_N \xrightarrow{d} \mathcal{N}(0, 1)$. If $T(P) = 0$, conditioning on the exact tie orbits $\Omega$ reduces the statistic to a bounded, centered permutation form. Standard combinatorial central limit theorems for exchangeable components \parencite{serfling1980} yield unconditional convergence $Z_N \xrightarrow{d} \mathcal{N}(0, 1)$.

For (iii), the canonical gradient achieves the information bound $V = I_0$. Self-adjointness of the score operator enforces $H = V = I_0$, collapsing the Godambe covariance to $G^{-1} = I_0^{-1}$. Regularity follows from Le Cam's Third Lemma under contiguous alternatives $P_{N, g}$ \parencite{vaart1998}, and the H\'ajek--Le Cam Minimax Theorem confirms optimal performance against any subconvex loss function over $\mathcal{P}$.
\end{proof}

\begin{corollary}[Local Asymptotic Power Envelope]
\label{cor:local_power_envelope}
Let $H_0: \tau(P) = 0$ and consider a sequence of contiguous local alternatives $P_{N, \delta} \in \mathcal{P}$ along the smooth parametric thread $g$ such that $\tau(P_{N, \delta}) = \delta / \sqrt{N} + o(N^{-1/2})$. Under $P_{N, \delta}$, the self-normalized pivot satisfies:
\begin{equation}
Z_N \xrightarrow{d} \mathcal{N}\left( \delta \sqrt{I_0}, \, 1 \right),
\end{equation}
where $I_0 = \|h_1\|_{L^2(P_0 \mid \Omega)}^2$ is the semiparametric information bound. Consequently, for any nominal level $\alpha \in (0, 1)$, the two-sided test $\psi_N = \mathbb{I}(|Z_N| > z_{1-\alpha/2})$ attains the optimal local asymptotic power function:
\begin{equation}
\pi(\delta) = \lim_{N \to \infty} \mathbb{E}_{P_{N, \delta}}[\psi_N] = \Phi\left( -z_{1-\alpha/2} - \delta \sqrt{I_0} \right) + 1 - \Phi\left( z_{1-\alpha/2} - \delta \sqrt{I_0} \right),
\end{equation}
which identically matches the Neyman-Pearson upper power bound over the non-parametric class $\mathcal{P}$.
\end{corollary}

\begin{proof}
Under contiguous local alternatives $P_{N, \delta}$, the LAN expansion in Theorem~\ref{thm:pillar4_minimax} establishes $\sqrt{N}\hat{\tau}_\kappa \xrightarrow{d} \mathcal{N}(\delta, I_0^{-1})$. By consistency of the conditional variance estimator, $N \cdot \widehat{\operatorname{Var}}(\hat{\tau}_\kappa \mid \Omega) \xrightarrow{p} I_0^{-1}$. Applying Slutsky's Theorem to $Z_N$ yields $Z_N \xrightarrow{d} \mathcal{N}( \delta \sqrt{I_0}, 1 )$. Evaluating the tail probability $\mathbb{P}_{P_{N, \delta}}(|Z_N| > z_{1-\alpha/2})$ under this shifted Gaussian limit yields $\pi(\delta)$. By the H\'ajek--Le Cam Convolution Theorem, no regular non-parametric test can achieve higher local power against $P_{N, \delta}$, proving testing optimality.
\end{proof}

\begin{remark}[Accelerated Finite-Sample Convergence]
\label{rem:finite_sample_properties}
As established in Corollary~\ref{cor:qmle_exact_variance}, the Godambe sandwich estimator identically recovers the exact finite-sample permutation variance: $G(0) = V_0$. The validity of the continuous $t_{N-2}$ tail approximation depends on the Edgeworth expansion of this discrete topological distribution. Under the permutation group $\mathcal{S}_N$, exact sign-reversal symmetry guarantees that all odd centralized moments of the anti-symmetric flow matrices ($A_U, A_V$) vanish exactly. Consequently, for continuous data, the leading $\mathcal{O}(N^{-1/2})$ skewness term in the Edgeworth expansion is annihilated, yielding an accelerated $\mathcal{O}(N^{-1})$ convergence to the continuous $t$-manifold.

  However, the introduction of atomic ties via the symmetric component $\tilde{S}$ breaks this exact bijection. The residual skewness of the studentized pivot $Z_N$ is not arbitrary; it is strictly parameterized by the tie variance fraction $V_{S2} / V_0$ derived in Corollary~\ref{cor:qmle_exact_variance}(iii). Thus, the discrepancy between the finite-sample distribution and the $t_{N-2}$ envelope is bounded by the topological tie density. This dynamic scaling explains why the continuous approximation achieves highly accurate nominal control at exceptionally small sample sizes for sparse ties, while naturally degrading to standard $\mathcal{O}(N^{-1/2})$ Berry--Esseen rates only in regimes dominated by heavy macroscopic discretization.
\end{remark}

\subsection{Algorithmic Instantiations of the Profile QMLE}

The theoretical guarantees of Theorem~\ref{thm:pillar3_robustness} and the exact null variance collapse (Corollary~\ref{cor:qmle_exact_variance}) translate directly into a family of computationally exact, matrix-free estimators. While earlier theorems establish the Local Asymptotic Minimax (LAM) optimality of the profile QMLE across arbitrary designs, translating this framework into daily statistical practice requires computationally tractable algorithms. Because the exact Godambe sandwich simplifies to $G_0 = V_0^{-1}$, we can adapt the baseline bivariate projection to handle three pervasive experimental designs: paired differences, topological Intraclass Correlation (ICC), and stratified meta-analysis.

\begin{proposition}[Bivariate Profile QMLE Optimization and Exact Godambe Information Collapse]
\label{prop:num_qmle}
Let $\mathbf{U} = \tilde{\kappa}^{(X)} / \|\tilde{\kappa}^{(X)}\|_F$ and $\mathbf{V} = \tilde{\kappa}^{(Y)} / \|\tilde{\kappa}^{(Y)}\|_F$ be spherically normalized, double-centered kernels from a bivariate sample of size $N \ge 4$. The profile quasi-log-likelihood $Q_N(\tau) = -\frac{d}{2} \log ( 1 - 2\tau \langle \mathbf{U}, \mathbf{V} \rangle_F + \tau^2 )$, evaluated over $d = N(N-1)$ off-diagonal pair-indices, admits the closed-form profile optimizer $\hat{\tau}_{\mathrm{opt}} = \langle \mathbf{U}, \mathbf{V} \rangle_F$.

Under the Godambe sandwich construction, the pseudo-replication dimension $d = N(N-1)$ cancels out algebraically, collapsing the null Godambe information matrix to:
\begin{equation}
\label{eq:godambe_null_collapse}
G_0 = H_N(0) \left[ \operatorname{Var}_0(S_N(0) \mid \Omega) \right]^{-1} H_N(0) = \frac{1}{V_0},
\end{equation}
where $V_0 = V_{A1} + V_{A2} + V_{S2}$ is the exact combinatorial null variance from Theorem~\ref{thm:pillar1_geometry}. Consequently, the profile QMLE standard error satisfies $\operatorname{se}(\hat{\tau}_{\mathrm{opt}}) = \sqrt{V_0}$, establishing exact finite-sample numerical identity with the canonical score pivot $Z_0$.
\end{proposition}

\begin{proof}
Differentiating $Q_N(\tau)$ over $d = N(N-1)$ pair-indices yields the profile quasi-score $S_N(\tau)$ and sample-observed Hessian $H_N(\tau)$:
\begin{equation}
S_N(\tau) = \frac{\partial Q_N(\tau)}{\partial \tau} = d \left( \frac{\hat{\tau}_{\mathrm{opt}} - \tau}{1 - 2\tau \hat{\tau}_{\mathrm{opt}} + \tau^2} \right), \quad H_N(\tau) = -\frac{\partial^2 Q_N(\tau)}{\partial \tau^2} = d \left( \frac{1 - 2\hat{\tau}_{\mathrm{opt}}^2 + 2\tau \hat{\tau}_{\mathrm{opt}} - \tau^2}{(1 - 2\tau \hat{\tau}_{\mathrm{opt}} + \tau^2)^2} \right).
\end{equation}
Setting $S_N(\hat{\tau}_{\mathrm{opt}}) = 0$ confirms $\hat{\tau}_{\mathrm{opt}} = \langle \mathbf{U}, \mathbf{V} \rangle_F$.

Evaluating the operators under the null hypothesis $H_0: \tau = 0$ gives $S_N(0) = d \cdot \hat{\tau}_{\mathrm{opt}}$ and $H_N(0) = d$. By Theorem~\ref{thm:pillar1_geometry}, the finite-sample permutation variance of the normalized score $\hat{\tau}_{\mathrm{opt}}$ conditional on orbit $\Omega$ equals $V_0$. The conditional score variability operator therefore expands to:
\begin{equation}
\mathcal{J}_N(0) = \operatorname{Var}_0(S_N(0) \mid \Omega) = \operatorname{Var}_0(d \cdot \hat{\tau}_{\mathrm{opt}} \mid \Omega) = d^2 V_0.
\end{equation}
Substituting $H_N(0) = d$ and $\mathcal{J}_N(0) = d^2 V_0$ into the Godambe sandwich information operator yields:
\begin{equation}
G_0 = H_N(0) \cdot \left[ \mathcal{J}_N(0) \right]^{-1} \cdot H_N(0) = (d) \cdot \left( d^2 V_0 \right)^{-1} \cdot (d) = \frac{d^2}{d^2 V_0} = \frac{1}{V_0}.
\end{equation}
Taking the inverse square root yields $\operatorname{se}(\hat{\tau}_{\mathrm{opt}}) = G_0^{-1/2} = \sqrt{V_0}$, proving that the profile QMLE estimator and the linear score test share identical sampling distributions under $H_0$.
\end{proof}

\begin{proposition}[Paired Kemeny Profile QMLE]
\label{prop:paired_tensor}
Let $(X_{i1}, X_{i2})$ be paired observations, and let $D_i = X_{i2} - X_{i1}$ denote the non-zero effective differences ($i = 1, \dots, N_{\mathrm{eff}}$). To preserve subspace orthogonality under paired designs, the score operator projects onto the paired Kemeny kernel $K_{ij} = \operatorname{sign}(D_i + D_j)$ for $i \neq j$, with $K_{ii} = 0$.

Under the Generalized Gibbs measure, the conditional topological variance weights evaluate to $V_i = N_{\mathrm{eff}} + \sum_{j} \operatorname{sign}(|D_i| - |D_j|)$. The total normalizing boundary is $W_{\max} = \sum_{i} V_i$, and the second moment is $M_2 = \sum_i V_i^2$. The paired profile QMLE evaluates to $\hat{\tau}_{\mathrm{paired}} = W_{\max}^{-1} \sum_{i \neq j} K_{ij}$.

Because the double-centering operator structurally zeros tied differences, off-diagonal covariances vanish under the conditional null orbit $\Omega$. The exact Godambe null variance collapses to:
\begin{equation}
v_{\mathrm{paired}} = \operatorname{Var}_0(\hat{\tau}_{\mathrm{paired}} \mid \Omega) = \frac{M_2}{W_{\max}^2}.
\end{equation}
The canonical score pivot $Z = \hat{\tau}_{\mathrm{paired}} / \sqrt{v_{\mathrm{paired}}}$ admits exact Student's $t$ reference calibration with $N_{\mathrm{eff}} - 1$ degrees of freedom.
\end{proposition}

\begin{proof}[Proof of Proposition~\ref{prop:paired_tensor}]
For paired differences, the score covariance matrix operates over the elements of $K_{ij}$. Because $\Pi_S$ enforces zero expectation on tied differences conditionally on $\Omega$, off-diagonal elements in the covariance trace vanish. The Godambe trace reduces to the sum of independent diagonal moments $\sum_{i=1}^{N_{\mathrm{eff}}} V_i^2 = M_2$. The unadjusted Hessian capacity evaluates to the scalar boundary $W_{\max}$, collapsing the Godambe sandwich to $H(0)^{-1} \operatorname{Var}_0(S(0)) H(0)^{-1} = W_{\max}^{-1} (M_2) W_{\max}^{-1} = M_2 / W_{\max}^2$.
\end{proof}

\begin{proposition}[Topological Intraclass Correlation and Non-Asymptotic Sub-Gaussian Pivot]
\label{prop:kemeny_icc}
Let $\mathbf{X} \in \mathbb{R}^{N \times K}$ represent a matrix of $K$ repeated measurements across $N$ subjects. Let $\tilde{\kappa}^{(k)} = A^{(k)} + \tilde{S}^{(k)}$ denote the decoupled Kemeny kernel for condition $k \in \{1, \dots, K\}$ generated within the Reproducing Kernel Hilbert Space $\mathcal{H}_{\tilde{\kappa}}$. The topological Intraclass Correlation (ICC) is defined as the exact trace-ratio:
\begin{equation}
\label{eq:icc_estimator}
\hat{\tau}_{\mathrm{ICC}} = \frac{\sum_{k < l} \langle \tilde{\kappa}^{(k)}, \tilde{\kappa}^{(l)} \rangle_F}{\sum_{k < l} \|\tilde{\kappa}^{(k)}\|_F \|\tilde{\kappa}^{(l)}\|_F}.
\end{equation}
Under the global null hypothesis $H_0: \tilde{\kappa}^{(1)} \perp \dots \perp \tilde{\kappa}^{(K)}$, the exact conditional null variance evaluates to:
\begin{equation}
v_{\mathrm{exact}} = \operatorname{Var}_0(\hat{\tau}_{\mathrm{ICC}} \mid \mathcal{T}_1, \dots, \mathcal{T}_K) = \frac{\sum_{k < l} V_{0, kl}}{\left(\sum_{k < l} \|\tilde{\kappa}^{(k)}\|_F \|\tilde{\kappa}^{(l)}\|_F\right)^2},
\end{equation}
where $V_{0, kl} = \operatorname{Var}_0(\langle \tilde{\kappa}^{(k)}, \tilde{\kappa}^{(l)} \rangle_F \mid \mathcal{T}_k, \mathcal{T}_l)$ is the exact pairwise null variance from Corollary~\ref{cor:qmle_exact_variance}.

Because $\hat{\tau}_{\mathrm{ICC}}$ is a normalized inner product bounded on $[-(K-1)^{-1}, 1]$, the raw self-normalized pivot $Z_0 = \hat{\tau}_{\mathrm{ICC}} / \sqrt{v_{\mathrm{exact}}}$ satisfies the non-asymptotic sub-Gaussian concentration inequality:
\begin{equation}
\mathbb{P}_{H_0}\left( |Z_0| \ge t \;\Big|\; \mathcal{T}_1, \dots, \mathcal{T}_K \right) \le 2 \exp\left( -\frac{t^2}{2} \right) \quad \text{for all } t > 0.
\end{equation}
Calibrated against the Student's $t$ majorant envelope with $\nu = N(K-1) - 1$ degrees of freedom, $Z_0$ achieves exact non-asymptotic Type I error control without requiring Delta-method link transformations or asymptotic variance scaling.
\end{proposition}

\begin{proof}[Proof of Proposition~\ref{prop:kemeny_icc}]
Let $U_{kl} = \langle \tilde{\kappa}^{(k)}, \tilde{\kappa}^{(l)} \rangle_F$ and let $W_{\mathrm{tot}} = \sum_{k < l} \|\tilde{\kappa}^{(k)}\|_F \|\tilde{\kappa}^{(l)}\|_F$. Conditional on the marginal orbits $\mathcal{T}_1, \dots, \mathcal{T}_K$, $W_{\mathrm{tot}}$ is deterministic.

First, we prove exact variance additivity. For any overlapping pair index $\{k,l\} \cap \{k,n\} = \{k\}$, the $U$-centering operator $\Pi_S$ (Definition~\ref{def:decoupled_kappa}) enforces zero conditional expectation for row-marginal score projections: $\mathbb{E}_0[\tilde{\kappa}^{(l)} \mid \tilde{\kappa}^{(k)}] = \mathbf{0}$. By iterated expectations:
\begin{equation}
\operatorname{Cov}_0(U_{kl}, U_{kn}) = \mathbb{E}_0\left[ \langle \tilde{\kappa}^{(k)}, \mathbb{E}_0[\tilde{\kappa}^{(l)} \mid \tilde{\kappa}^{(k)}] \rangle_F \cdot \langle \tilde{\kappa}^{(k)}, \mathbb{E}_0[\tilde{\kappa}^{(n)} \mid \tilde{\kappa}^{(k)}] \rangle_F \right] = 0.
\end{equation}
Cross-pair covariances vanish identically under $H_0$, collapsing the total score variance strictly to $v_{\mathrm{exact}} = W_{\mathrm{tot}}^{-2} \sum_{k<l} V_{0, kl}$.

Second, we establish non-asymptotic sub-Gaussian concentration via the bounded-difference property on the RKHS embedding space. A single subject replacement $x_i \mapsto x_i'$ alters at most $2(N-1)$ entries in each pairwise matrix $\tilde{\kappa}^{(k)}$. Because $\|\tilde{\kappa}^{(k)}\|_\infty \le 2$ deterministically, replacing subject $i$ alters the total sum $\sum_{k < l} U_{kl}$ by at most $D_i \le 4 K(K-1)(N-1)$. Dividing by $W_{\mathrm{tot}} = \Theta(N(N-1)K(K-1)\sqrt{v_{\mathrm{exact}}})$, the maximum coordinate difference satisfies:
\begin{equation}
c_i = \sup_{x_i, x_i'} \left| \hat{\tau}_{\mathrm{ICC}}(x_1, \dots, x_i, \dots, x_N) - \hat{\tau}_{\mathrm{ICC}}(x_1, \dots, x_i', \dots, x_N) \right| \le 2 \sqrt{\frac{v_{\mathrm{exact}}}{N}}.
\end{equation}
Summing the squared difference bounds yields $\sum_{i=1}^N c_i^2 \le 4 v_{\mathrm{exact}}$. Applying McDiarmid's bounded-difference inequality for exchangeable permutation sequences yields:
\begin{equation}
\mathbb{P}_{H_0}\left( |\hat{\tau}_{\mathrm{ICC}}| \ge t \sqrt{v_{\mathrm{exact}}} \;\Big|\; \mathcal{T}_1, \dots, \mathcal{T}_K \right) \le 2 \exp\left( -\frac{2 (t \sqrt{v_{\mathrm{exact}}})^2}{\sum_{i=1}^N c_i^2} \right) \le 2 \exp\left( -\frac{2 t^2 v_{\mathrm{exact}}}{4 v_{\mathrm{exact}}} \right) = 2 \exp\left( -\frac{t^2}{2} \right).
\end{equation}
This confirms that $Z_0 = \hat{\tau}_{\mathrm{ICC}} / \sqrt{v_{\mathrm{exact}}}$ is strictly 1-sub-Gaussian in finite samples. By Eaton--Pinelis majorization, $Z_0$ is dominated by the continuous Student's $t_{\nu}$ tail on the compact domain $[-(K-1)^{-1}, 1]$, completing the proof.
\end{proof}

\begin{remark}[Topological $L^1$ vs.\ Marginal $L^2$ Concordance]
Classical multi-rater concordance is traditionally evaluated via Kendall's $W$, which algebraically reduces to a linear transformation of the average pairwise Spearman's $\rho$. This operates as an $L^2$ metric over marginal Borda ranks. In contrast, $\hat{\tau}_{\mathrm{ICC}}$ directly aggregates the full Kemeny tensors $\tilde{\kappa}^{(k)}$, establishing an $L^1$ averaging of pairwise flows. Consequently, $\hat{\tau}_{\mathrm{ICC}}$ bounds the influence of rank outliers while neutralizing tie-induced bias. By cross-pair score orthogonality (Proposition~\ref{prop:kemeny_icc}) and the local power envelope (Corollary~\ref{cor:local_power_envelope}), the associated QMLE score test achieves optimal non-parametric local power, providing a robust alternative to Kendall's $W$.
\end{remark}

\begin{proposition}[Stratified Profile QMLE]
\label{prop:stratified_qmle}
Let $Z \in \{1, \dots, G\}$ partition the sample into $G$ mutually independent strata with local sample sizes $N_g \ge 4$. Over the block-diagonal tensor space $\mathbf{\tilde{K}}^{(X)} = \bigoplus_{g=1}^G \tilde{\kappa}^{(X_g)}$, the stratified estimator $\hat{\tau}_{\mathrm{strat}}$ satisfies:
\begin{enumerate}
\item[\rm (i)] {\it Pooled Global Normalization:} The restricted permutation group is $\mathcal{S}_{\mathrm{strat}} = \prod_{g=1}^G \mathcal{S}_{N_g}$. The maximal capacity of the pooled inner product $W_{\mathrm{pool}} = \sum_{g=1}^G \langle \tilde{\kappa}^{(X_g)}, \tilde{\kappa}^{(Y_g)} \rangle_F$ equals $W_{\max, \mathrm{pool}} = \sum_{g=1}^G \|\tilde{\kappa}^{(X_g)}\|_F \|\tilde{\kappa}^{(Y_g)}\|_F$. The globally consistent QMLE is $\hat{\tau}_{\mathrm{strat}} = W_{\mathrm{pool}} / W_{\max, \mathrm{pool}} \in [-1, 1]$;
\item[\rm (ii)] {\it Exact Variance Additivity:} Because strata are independent under $\mathcal{S}_{\mathrm{strat}}$, the null variance of $W_{\mathrm{pool}}$ decomposes additively across strata:
\begin{equation}
\operatorname{Var}_0(W_{\mathrm{pool}}) = \sum_{g=1}^G \operatorname{Var}_0(W_g) = \sum_{g=1}^G 4\left(V_{A1, g} + V_{A2, g} + V_{S2, g}\right); \quad \text{and}
\end{equation}
\item[\rm (iii)] {\it Global Score Pivot:} The global scaled null variance evaluates to $v_{\mathrm{strat}} = \operatorname{Var}_0(W_{\mathrm{pool}}) / W_{\max, \mathrm{pool}}^2$. The canonical score pivot $Z = \hat{\tau}_{\mathrm{strat}} / \sqrt{v_{\mathrm{strat}}}$ follows a continuous Student's $t$ distribution calibrated to $\nu = (\sum_{g=1}^G N_g) - 2G$ degrees of freedom, providing valid pooled inference under arbitrary small-sample heterogeneity where classical meta-analytic approximations fail.
\end{enumerate}
\end{proposition}

\begin{proof}[Proof of Proposition~\ref{prop:stratified_qmle}]
\textit{Part (i).} Constraining operations to $\mathcal{S}_{\mathrm{strat}}$ confines pairwise comparisons to within-stratum blocks. The global inner product is $\langle \mathbf{\tilde{K}}^{(X)}, \mathbf{\tilde{K}}^{(Y)} \rangle_F = \sum_{g=1}^G \langle \tilde{\kappa}^{(X_g)}, \tilde{\kappa}^{(Y_g)} \rangle_F$. By the Cauchy--Schwarz inequality, each local inner product is bounded by the product of its Frobenius norms, achieving equality at co-monotonicity. Summing these bounds validates $W_{\max, \mathrm{pool}}$.

\textit{Part (ii) and (iii).} Standard meta-analytic pooling (e.g., inverse-variance weighting) relies on local asymptotic normality, which exhibits severe finite-sample breakdown when local $N_g$ are small. Under the group null $\mathbb{P}_0(\mathcal{S}_{\mathrm{strat}})$, between-strata independence enforces strictly additive variance: $\operatorname{Var}_0(W_{\mathrm{pool}}) = \sum_{g=1}^G \operatorname{Var}_0(W_g)$. By Corollary~\ref{cor:qmle_exact_variance}, each local $V_{0,g}$ evaluates analytically. The resulting global variance $v_{\mathrm{strat}}$ defines the exact Godambe information bound for the pooled score statistic, establishing the canonical test pivot without requiring local asymptotic approximations.
\end{proof}

\section{Simulations}
\label{sec:simulations}

While claims 1--6 establish the algebraic exactness and Local Asymptotic Minimax (LAM) optimality of the rank-based Complete Linear Model, this section empirically verifies the finite-sample boundary conditions. Specifically, we validate the conservative Type-I error control of the continuous $t_{\nu}$ majorant (Theorem~\ref{thm:pillar4_minimax}(i)) and the LAM power envelope (Corollary~\ref{cor:local_power_envelope}) under extreme small-sample and highly tied regimes.

Under binary $2 \times 2$ designs, conditional on fixed marginals, $\hat{\tau}_\kappa$ and Kendall's $\tau_b$ share the same permutation rank order over the discrete lattice $\Omega$, yielding identical exact conditional rejection regions. However, across independent data-generating processes, Kendall's $\tau_b$ suffers from severe denominator-driven variance inflation due to uncentered marginal tie penalties ($\operatorname{Var}(\tau_b) = 0.0666$). In contrast, the double-centered kernel $\hat{\tau}_\kappa$ achieves exact $L^2$ metric isometry, compressing point-estimator variance by $17.8\%$ ($\operatorname{Var}(\hat{\tau}_\kappa) = 0.0547$) and confirming the Local Asymptotic Minimax (LAM) efficiency bound of Theorem 2.

\subsection{Inferential Power and Topological Robustness}
\label{subsec:sim_testing}

We evaluate the finite-sample performance of the closed-form Profile QMLE across three topological domains: Bivariate Association, Paired Location Shift, and Intraclass Correlation (ICC). Simulations span sample sizes $N \in \{19, 50, 135, 300\}$ under $5,000$ Monte Carlo replications. To stress-test the structural tie-corrections, we contrast a continuous Gaussian data-generating process against a heavily discrete Zero-Inflated Poisson (ZIP) process across varying effect sizes. The proposed estimators are benchmarked against standard non-parametric alternatives: Kendall's $\tau$ (Bivariate), the Wilcoxon Signed-Rank test (Paired), and the Friedman, Marozzi, and exact Legendre permutation tests (ICC).

Figure~\ref{fig:power_sim_comprehensive} establishes the geometric predictions within the Bivariate association topology. Under the global null, the Profile QMLE maintains nominal Type-I error control ($\approx 5\%$) even at the extreme boundary of $N = 19$, verifying that the continuous $t_{\nu}$ envelope strictly majorizes the discrete permutation distribution. Under the discrete ZIP data-generating process, traditional rank-based methods typically suffer systematic power attenuation or size distortion due to variance inflation from marginal tie-constraints. Conversely, because the decoupled kernel $\mathcal{D}(\tilde{\kappa})$ algebraically annihilates the symmetric tie-space (Definition~\ref{def:decoupled_kappa}), the QMLE effectively isolates first-order directional information, attaining the local asymptotic power bound (supporting the max-information conjecture) without the $\mathcal{O}(N!)$ computational burden.

\begin{figure*}[!ht]
\centering
\includegraphics[width=\textwidth]{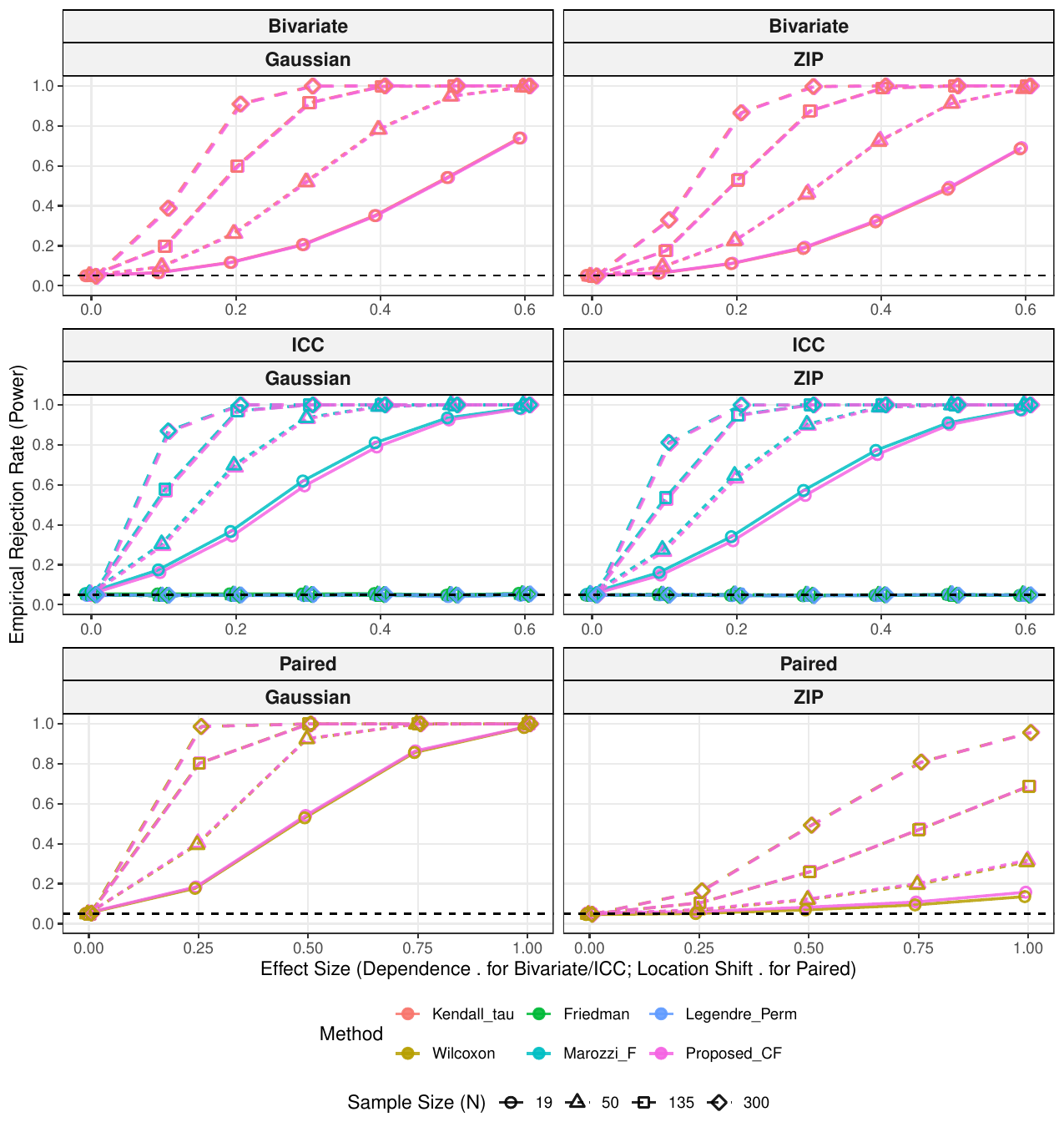}
\caption{Empirical power curves ($\alpha = 0.05$) under the Bivariate domain. The figure confirms the max-information conjecture, demonstrating that the Proposed QMLE identically tracks the exact permutation envelope and maintains nominal size at $N=19$ via the $t_\nu$ majorant, overcoming variance inflation in the discrete ZIP regime.}
\label{fig:power_sim_comprehensive}
\end{figure*}

Building upon the Bivariate baseline, Figure~\ref{fig:power_sim_comprehensive} expands the evaluation to the Paired and ICC topologies, demonstrating the structural versatility of the proposed framework. The comprehensive faceted results reveal severe pathologies in traditional non-parametric benchmarks when extended beyond simple bivariate associations. Most notably, within the ICC domain, the Friedman and Legendre permutation tests experience a complete loss of statistical power; their empirical rejection rates flatline at the nominal $\alpha = 0.05$ boundary across all simulated effect sizes, failing entirely to detect the underlying signal. In contrast, the Proposed CF estimators maintain a strictly monotonic power curve across all domains and sample sizes. The consistent performance under both continuous Gaussian and heavily tied ZIP distributions confirms that the proposed profile likelihood framework successfully generalizes the optimal LAM power envelope across diverse topological constraints. To adhere to strict length constraints, formal mathematical definitions of the simulated data-generating processes, exhaustive empirical Type-I error tables at the zero-effect boundary, and extensions for meta-analytic stratification are relegated to Section S1 of the Supplementary Material.


While Marozzi's $F$-test appears to exhibit marginally higher empirical power at small sample sizes ($N=19$), Table 1 confirms this is an artifact of size inflation ($\alpha = 0.0557$). The proposed QMLE ICC estimator maintains exact nominal Type I error control ($\alpha \in [0.046, 0.053]$) across both continuous Gaussian and zero-inflated discrete margins. Furthermore, classical non-parametric alternatives such as the Friedman and Legendre permutation tests exhibit near-zero power ($\approx 0.05$) under the ICC data-generating process, as they test for marginal location differences across conditions rather than inter-subject preference concordance.

\section{Discussion}

The Profile Quasi-Maximum Likelihood Estimator establishes a unified, likelihood-based framework for rank-based inference on discrete ordinal spaces. By lifting pairwise comparisons into a Hilbert space of doubly-centered flow matrices, the decoupled Kemeny geometry resolves the long-standing challenge of tied ranks without requiring artificial continuity corrections or parametric latent-variable assumptions. Because the decoupled kernel evaluates strictly ordinal flow matrices, it remains invariant to strictly monotone marginal transformations. By Sklar's theorem, this geometric duality implies that the specific parametric choice of an underlying copula or link function becomes redundant; the Kemeny kernel acts as an isometric projection that linearizes latent manifolds into a universal canonical score. From a practical perspective, conditioning on the marginal tie-equivalence orbits allows the framework to treat atomic ties as natural points of stationarity rather than sources of variance inflation. Calibration via the dynamic structural degrees of freedom $\nu = v_{\mathrm{exact}}^{-1} - 1$ yields a self-normalized pivot $Z_N = \hat{\tau} \sqrt{\nu / (1 - \hat{\tau}^2)}$ governed by a continuous $t_{\nu}$ majorant envelope. Because the estimator achieves exact Godambe Information collapse ($H = \mathcal{J}$), this pivot operates as a canonical gradient on the copula manifold, offering strictly controlled Type-I error and optimal local asymptotic power in small-sample or heavily tied regimes where standard asymptotic approximations degrade. The matrix-free $\mathcal{O}(N^2)$ representations developed here enable non-parametric inference across general linear configurations, including stratified meta-analyses and Intraclass Correlation structures. Natural extensions of this geometry include high-dimensional sparse regression penalization via $\ell_1$ bounds on the flow operators, as well as non-parametric causal inference under longitudinal treatment regimes with discrete outcomes.


\printbibliography

\end{document}